\documentclass[11pt, reqno]{amsart}
\usepackage{fullpage}
\usepackage[OT1]{fontenc}
\usepackage{amsthm,amsmath,amsfonts,amssymb,amsaddr}
\usepackage{hyperref}
\usepackage{xcolor}
\hypersetup{
    colorlinks=true,
    linkcolor=purple,    
    urlcolor=blue,
    citecolor=blue
}

\usepackage{mathtools}
\usepackage{bm}
\usepackage{fancybox}
\usepackage{graphicx}
\usepackage{xcolor}
\usepackage{booktabs}
\usepackage{multirow}
\usepackage{subcaption}
\usepackage[round]{natbib}
\usepackage[ruled, vlined]{algorithm2e}
\SetAlgoCaptionLayout{centerline}

\usepackage{comment}

\usepackage{pifont}
\newcommand{\cmark}{\ding{52}}
\newcommand{\xmark}{\ding{56}}

\newtheorem{proposition}{Proposition}

\newtheorem{assumption}{Assumption}
\theoremstyle{definition}

\newtheorem{remark}{Remark}

\def\T{{ \top }}

\newcommand{\ind}{\stackrel{\mathclap{\normalfont\text{\tiny ind}}}{\sim}}
\newcommand{\given}{\mid}
\newcommand{\biggiven}{\,\middle|\,}

\newcommand{\Real}{\mathbb{R}}
\newcommand{\calD}{\mathcal{D}}
\newcommand{\calL}{\mathcal{L}}
\newcommand{\calS}{\mathcal{S}}
\newcommand{\lTilde}{\tilde{\ell}}
\newcommand{\ITilde}{\tilde{I}}
\newcommand{\Borel}{\mathfrak{B}}
\newcommand{\BigO}{\mathcal{O}}
\newcommand{\calT}{\mathcal{T}}
\newcommand{\calM}{\mathcal{M}}
\newcommand{\GP}{\mathsf{GP}}
\newcommand{\Norm}{\mathsf{N}}
\newcommand{\IG}{\mathsf{IG}}
\newcommand{\tdist}{\mathsf{t}}
\newcommand{\chol}{\mathsf{chol}}

\title[Bayesian predictive stacking for spatially-temporally misaligned data]{Bayesian inference for spatially-temporally misaligned data using predictive stacking}

\author{Soumyakanti Pan$^{\star}$ \and Sudipto Banerjee$^{\star}$}
\address{$^\star$Department of Biostatistics, University of California Los Angeles}
\email{span18@ucla.edu, sudipto@ucla.edu}

\date{\today}

\begin{document}
\begin{abstract}
    Air pollution remains a major environmental risk factor that is often associated with adverse health outcomes. However, quantifying and evaluating its effects on human health is challenging due to the complex nature of exposure data. Recent technological advances have led to the collection of various indicators of air pollution at increasingly high spatial-temporal resolutions (e.g., daily averages of pollutant levels at spatial locations referenced by latitude-longitude). However, health outcomes are typically aggregated over several spatial-temporal coordinates (e.g., annual prevalence for a county) to comply with survey regulations. This article develops a Bayesian hierarchical model to analyze such spatially-temporally misaligned exposure and health outcome data. We introduce Bayesian predictive stacking, which optimally combines multiple predictive spatial-temporal models and avoids iterative estimation algorithms such as Markov chain Monte Carlo that struggle due to convergence issues inflicted by the presence of weakly identified parameters. We apply our proposed method to study the effects of ozone on asthma in the state of California.
    
\end{abstract}
\maketitle

\section{Introduction}
Spatial and temporal misalignment refers to the setting in which different variables are observed over incompatible spatial supports and/or at asynchronous time points or intervals. To be more specific, spatial misalignment arises when variables are measured at different spatial resolutions, i.e., different sets of locations or areas. For example, a variable might be observed at points (e.g. air quality monitoring stations), which we call \emph{point-referenced} data, while another is aggregated over administrative units (e.g., counties or census tracts), which we call \emph{block} data. Similarly, temporal misalignment occurs when variables are recorded on different time scales. 

In this article, we devise a Bayesian hierarchical modeling framework to study the effects of ozone on asthma-related health emergencies among California residents. Several studies collectively underscore the significant impact of exposure to ozone on asthma-related emergency department visits \citep[for e.g.,][]{Gharibi2019, NASSIKAS2020}. We obtain data on the monthly average concentration of ozone in California measured by air quality monitoring stations located at different locations throughout the state. Thus, the ozone measurements are spatially point-referenced and temporally aggregated at a monthly scale. However, data on asthma-related emergency department visits among Californian residents are reported at the county level and aggregated annually. We refer to such disparity in both spatial and temporal resolutions as spatial-temporal misalignment. 

Spatially-temporally misaligned data presents significant challenges for coherent modeling, prediction, and inference. This problem is well known as the \emph{change of support} and \emph{modifiable areal unit} problem \citep{Cressie1996, Mugglin2000, gelfand2001_biostat, Gotway2002}. Most of the literature focuses on estimating variables at unobserved spatial resolutions or integrating data across varying spatial scales \citep{banerjee_spatial, zhongMoraga2023}. However, efforts to estimate the association between an outcome of interest and a spatially temporally misaligned covariate remain rather limited \citep{ZhuCarlinGelfand2003, cameletti2019}, with existing approaches typically overlooking temporal misalignment. Moreover, a key challenge in analyzing spatial-temporal exposure data is missing observations resulting from intermittent monitoring or data removal due to quality issues such as sensor failures. Traditional approaches often use imputation methods that aim to reconstruct the complete data set \citep[for e.g.,][]{ZhuCarlinGelfand2003, quick2013aoas}. Our approach forgoes imputation and works directly with the available irregularly spaced data to produce a more robust framework for spatial-temporal analysis.

We contribute in two novel aspects. First, we propose a modular Bayesian inference framework \citep{Bayarri2009, jacob2017} that regresses an outcome on a spatially-temporally misaligned covariate, based on noisy observations of the latter, ensuring fully model-based propagation of inferential uncertainty. Second, we develop predictive stacking for estimation of such models that analyze spatially-temporally misaligned data, thus representing a methodological advancement over previous work that included spatial Gaussian data \citep{zhang2024stacking, wakayama2024, presicce2025}, and spatial-temporal non-Gaussian data \citep{pan2025stacking}. Stacking \citep{wolpert1992stacked, breiman1996stacked, clyde2013bayesian} is conspicuous in machine learning as an effective alternative \citep[][]{le2017bayes, yao2018using} to traditional Bayesian model averaging \citep{madigan1996bayesian, hoeting1999bma}. The underlying idea in predictive stacking is to optimally assimilate posterior distributions on a grid of candidate values corresponding to intractable and weakly identified hyperparameters, such as the spatial-temporal decay, smoothness, and the noise-to-spatial-temporal variance ratio \citep{Zhang2004, ZhangZimmerman2005, tangEtAl2021}, which impede the convergence of iterative algorithms such as Markov chain Monte Carlo (MCMC). Stacking also differs from quadrature-based approaches, such as INLA \citep{inla2009}, in that we avoid approximating the posterior distribution of weakly identified parameters. Instead, we average, or ``stack'' individual posterior distributions, using weights obtained by optimizing a proper scoring rule \citep{Gneiting2007}.

The remainder of the manuscript is structured as follows. Section~\ref{sec:data-description} describes the data set that motivates our methodology. Section~\ref{sec:bayesianmodel} introduces our Bayesian hierarchical model for analyzing the spatially-temporally misaligned data and states model assumptions that are critical for posterior inference. Section~\ref{sec:stacking} develops predictive stacking and focuses on developing a computationally efficient algorithm for model estimation. Section~\ref{sec:data-analysis} describes the detailed analysis of our data set, while Section~\ref{sec:discussion} concludes with a brief discussion.

\section{Data}\label{sec:data-description}
\begin{figure}[t]
    \centering
    \includegraphics[width=0.95\linewidth]{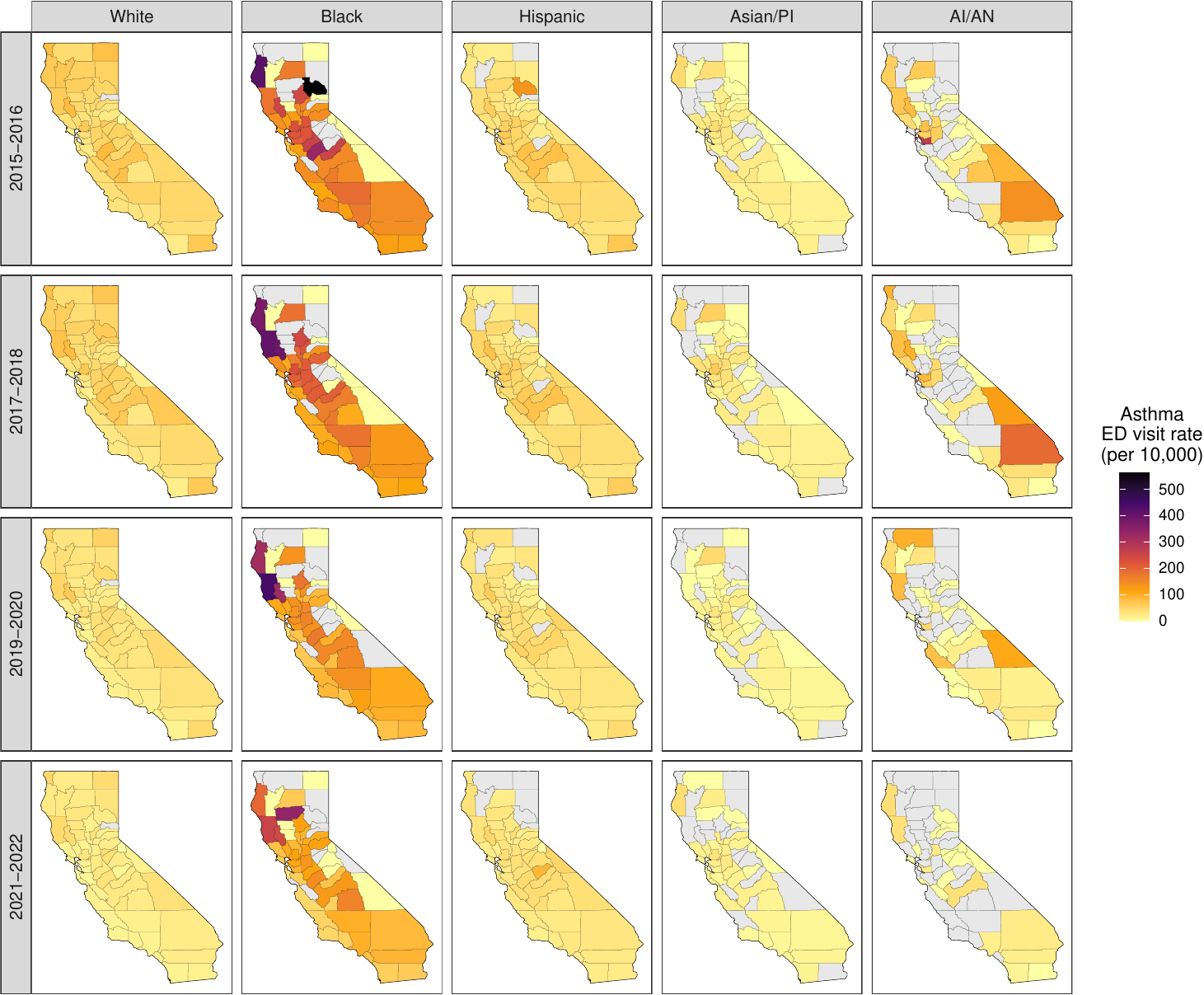}
    \caption{Biennial average asthma-related emergency department visit rates (per 10,000) by racial group for each California county from 2015 through 2022. For visualization, rates are averaged over consecutive 2-year periods.}
    \label{fig:asthma-EDV}
\end{figure}
The data set on adverse health outcomes comprises annual county-level rates of visits to the emergency department (ED) related to asthma per 10,000 residents of California. These data are obtained from the California Department of Public Health (CDPH) and are derived from the Department of Health Care Access and Information's Emergency Department database, which includes records from all licensed hospitals in the state \citep{asthma-CDPH}. We analyze data consisting of age-adjusted rates stratified by race/ethnicity (white, black, Hispanic, Asian/Pacific Islander, American Indian/Alaskan Native), derived from annual aggregated counts of asthma-related ED visits from 2015 to 2022 for each of the 58 counties in California. 
ED visit counts are based only on primary discharge diagnosis codes. The database omits rates based on (1) counts $< 12$ due to statistical instability, and (2) counts $\leq 5$ according to the CalHHS Data De-identification Guidelines. This results in approximately 35\% missing data 
that mainly affect the Asian/Pacific Islander and American Indian/Alaskan Native groups. Figure~\ref{fig:asthma-EDV} maps the aggregated two-year rates of asthma-related emergency department visits observed for each county in California during 2015-2022, revealing a clear pattern of racial disparity, with some counties consistently showing high rates. 
The age-adjusted rates are highly positive skewed, ranging from 0 to 560.9 per 10,000 residents, with the majority between 15 and 60 (see Figure~\ref{fig:asthma-hist} in the Appendix). Hence we analyze log-transformed rates, as they are more amenable to Gaussian assumptions, that yields analytically tractable posteriors, enabling inference without iterative algorithms.

For exposure data, we extracted hourly measurements of ozone concentration (in parts per million) from the California Air Quality and Meteorological Information System (AQMIS) database of the California Air Resources Board for the years 2015-2022. These data were recorded from about 200 air quality monitoring sites in California that were active during this time period. Figure~\ref{fig:ca-county} maps the geographic locations of these monitoring sites located in the county boundaries of California, highlighting the heterogeneous spatial distribution of air quality surveillance in counties. 
\begin{figure}
    \centering
    \includegraphics[width=0.6\linewidth]{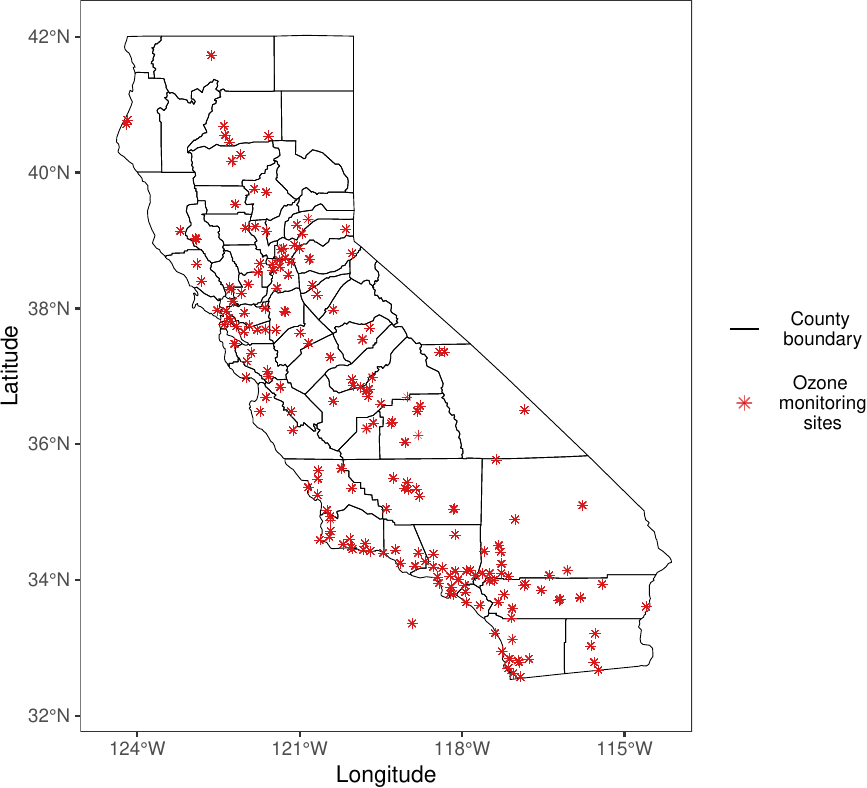}
    \caption{County boundaries of California and the geographic locations of 200 ozone monitoring sites active from 2015 to 2022.}
    \label{fig:ca-county}
\end{figure}
The figure also reveals a clear regional clustering of sites near urban and coastal areas in central and southern California, with notable sparsity in the eastern rural inland regions of the state. For our analysis, we aggregate the hourly data into monthly average ozone concentrations. Due to the fact that not all monitoring sites were active throughout the study period and many have missing monthly records, the resulting data set is temporally unbalanced between sites. In total, the data set comprises measurements of ozone concentrations at 15,725 unique spatial-temporal locations, where the temporal component is referenced by monthly intervals rather than exact timestamps. Figure~\ref{fig:oz-surface} features interpolated spatial surfaces of annual average ozone concentration between 2015 and 2022, showing a clear spatial pattern, with the sparsely sampled eastern inland regions recording higher ozone concentrations, demonstrating spatial imbalances in monitoring coverage.
\begin{figure}
    \centering
    \includegraphics[width=0.95\linewidth]{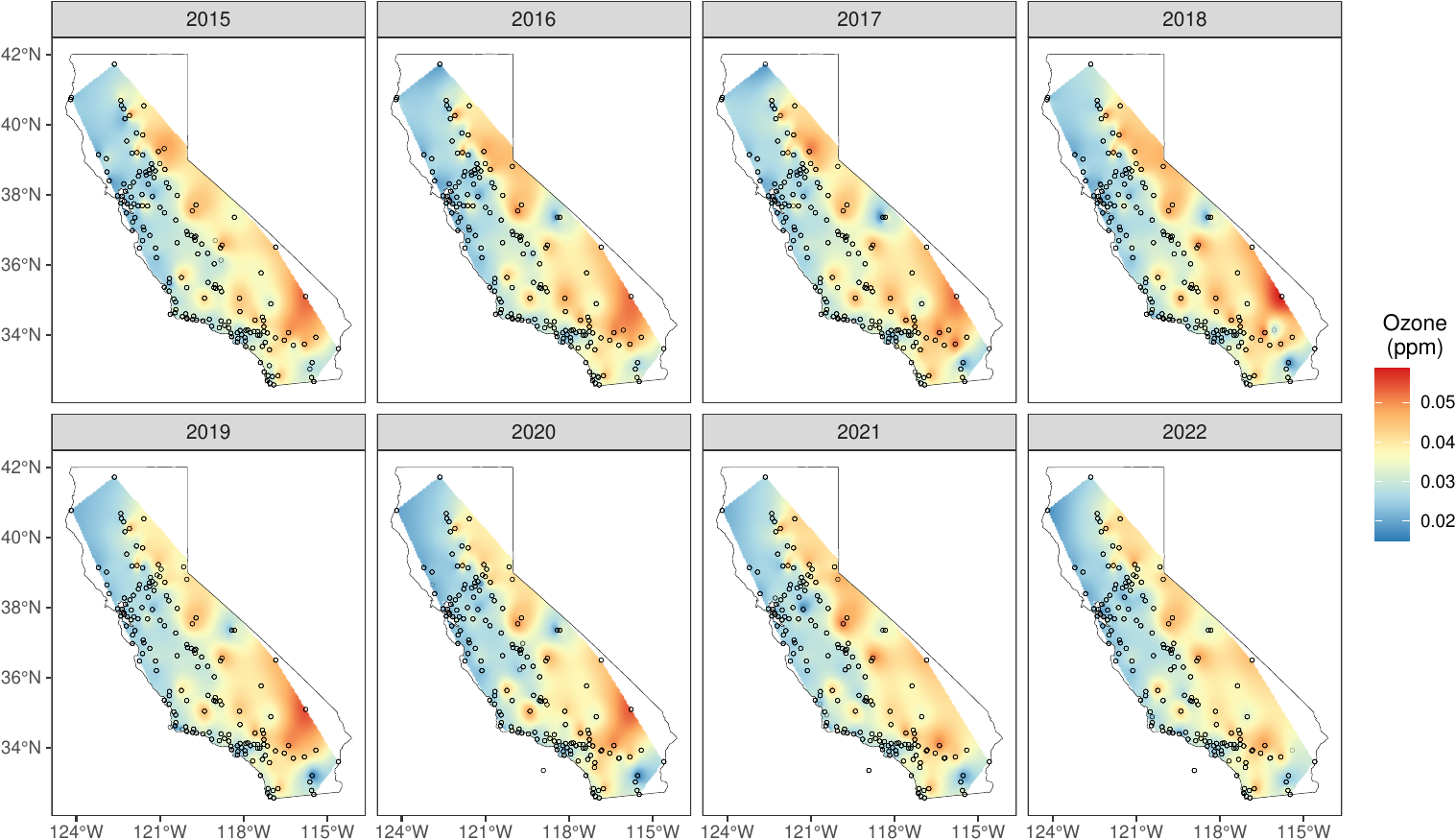}
    \caption{Interpolated spatial surface of annual (summed over months) ozone concentration (in parts per million) for California from 2015 to 2022. The geographic coordinates of the air quality monitoring stations are marked by black circles.}
    \label{fig:oz-surface}
\end{figure}

\section{Bayesian spatial-temporal hierarchical model}\label{sec:bayesianmodel}
\subsection{Multi-resolution spatial-temporal process}\label{subsec:spt-process}
We consider a spatial-temporal process as an uncountable set of random variables, say $\{Z(\ell): \ell\in \calD\}$, which is endowed with a probability law specifying the joint distribution for any finite sample of locations in $\calD = \calS\times \calT$, where $\calS \subset \Real^2$ and $\calT \subset [0,\infty)$ are space and time domains, respectively, and $\ell = (s,t)$ is a spacetime coordinate with $s \in \calS$ and $t\in \calT$ \citep[see, e.g.,][]{gnei10}. Subsequently, we define two new types of space-time coordinates $\lTilde = (\tilde{s}, \tilde{I})$ and $L = (B, I)$, where $\tilde{s} \in \calS$, $B \subset \calS$ denotes a block or a region within $\calS$, and $I, \tilde{I} \subset \calT$ denote intervals in $\calT$. We extend the spatial-temporal process $Z(\ell)$ to these new coordinate systems as $\{Z(\lTilde): \lTilde \in \calS \times \Borel(\calT) \}$ and $\{Z(L): L \in \Borel(\calS \times \calT) \}$, where for any set $A$, $\Borel(A)$ denotes the collection of Borel-measurable subsets of $A$ \citep{billingsley1995}, and
\begin{equation}\label{eq:process-def}
\begin{split}
    Z(\lTilde) = Z(\tilde{s}, \tilde{I}) = |\tilde{I}|^{-1} \int_{\tilde{I}} Z(\tilde{s}, t) dt\;,\quad
    Z(L) = Z(B, I) = (|B| |I|)^{-1} \int_I \int_B Z(s, t) ds dt\;,
\end{split}
\end{equation}
where $|B| = \int_B 1ds$ denotes the area of $B$, and $|I| = \int_{I} 1 dt$, $|\tilde{I}| = \int_{\tilde{I}} 1 dt$ denote the lengths of the intervals $I$ and $\tilde{I}$, respectively. This implies that the random variable $Z(\lTilde)$ denotes a realization of the stochastic process obtained by a temporal averaging of $Z(\ell)$ over the interval $\tilde{I}$ at location $\tilde{s}$. In our context, this corresponds to the monthly average concentration of ozone at the location $\tilde{s}$, averaged over the month $\tilde{I}$. Similarly, $Z(L)$ is a realization from the stochastic process obtained by averaging $Z(\ell)$ over a spatial-temporal block $L$, which includes averaging spatially over the region $B$ and averaging temporally over the interval $I$. This may represent ozone concentrations in the county defined by the spatial block $B$, averaged over the annual interval $I$.

A stationary Gaussian process specification for $Z(\ell)$ enables us to analytically derive a joint distribution for the spatial-temporal process at each resolution. To be specific, suppose $Z(\ell)$ is a Gaussian process with mean function $\mu(\ell; \gamma)$ and covariance function $\sigma^2 C(\ell, \ell'; \phi)$, where $\mu(\cdot; \gamma)$ denotes a trend surface with coefficient vector $\gamma$, $\sigma^2$ denotes the spatial-temporal variance, and $\phi$ denotes a generic parameter vector characterizing spatial-temporal decay or smoothness. We collectively call $\phi$ as ``process parameters''. For any given spatial-temporal coordinates 
$\{\lTilde_j = (\tilde{s}_j, \tilde{I}_j) : j = 1, \ldots, N \}$ and $\{L_k = (B_k, I_k) : k = 1, \ldots, K \}$, define
the $N \times 1$ vector $Z_{\lTilde} = (Z(\lTilde_1), \ldots, Z(\lTilde_N))^{\T}$, and the $K \times 1$ vector $Z_{L} = (Z(L_1), \ldots, Z(L_K))^{\T}$. Then, we have
\begin{equation}\label{eq:jointGP}
\begin{split}
    p \left( \begin{pmatrix} Z_{\lTilde}\\ Z_{L} \end{pmatrix} \biggiven \gamma, \sigma^2, \phi \right) = \Norm \left( \begin{pmatrix} Z_{\lTilde}\\ Z_{L} \end{pmatrix} \biggiven \begin{pmatrix} \mu_{\lTilde}(\gamma)\\ \mu_L(\gamma) \end{pmatrix}, \sigma^2 
    \begin{bmatrix}
        C_{\lTilde}(\phi) & C_{\lTilde, L}(\phi)\\
        C_{\lTilde, L}(\phi)^{\T} & C_{L}(\phi)
    \end{bmatrix}
    \right) \;,
\end{split}
\end{equation}
where 
\begin{equation}\label{eq:covs}
\begin{split}
    (\mu_{\lTilde}(\gamma))_j = |\tilde{I}_j|^{-1} \int_{\ITilde_j} &\mu((\tilde{s}_j, t); \gamma) dt, \quad (\mu_{L}(\gamma))_k = (|I_k| |B_k|)^{-1} \int_{L_k} \mu(\ell; \gamma) d\ell\;,\\
    (C_{\lTilde}(\phi))_{j, j'} &= (|\ITilde_j| |\ITilde_{j'}|)^{-1} \int_{\ITilde_{j'}} \int_{\ITilde_j} C((\tilde{s}_j, t), (\tilde{s}_{j'}, t'); \phi) dt dt' \;, \\
    (C_{\lTilde, L}(\phi))_{j, k} &= (|\ITilde_j| |I_k| |B_k|)^{-1} \int_{\ITilde_j} \int_{I_k} \int_{B_k} C((\tilde{s}_j, t'), (s, t); \phi) ds dt dt' \;, \\
    (C_{L}(\phi))_{k, k'} &= (|L_{k'}| |L_k|)^{-1} \int_{L_{k'}}\int_{L_k} C(\ell, \ell'; \phi) d\ell d\ell' \;,
\end{split}
\end{equation}
for $j, j' = 1, \ldots, N$, and $k, k' = 1, \ldots, K$, and $|L_k| = |B_k||I_k|$ for each $k$. The joint distribution \eqref{eq:jointGP} provides a unified framework that connects the process across the two spatial-temporal resolutions, enabling tractable posterior predictive inference at spatial-temporal blocks $L_k$ for each $k$, conditional on observed realizations at point-referenced temporal blocks $\lTilde_j, j = 1, \ldots, N$.

\subsection{Conjugate Bayesian hierarchical model}\label{subsec:hier}
Let $\calL = \{L_k: k = 1, \ldots, K\}$ be a fixed set of $K$ spatial-temporal blocks in $\Borel(\calS \times \calT)$, where $L_k$ is of the form $(B_k, I_k)$ with $B_k$ being a block or region in $\calS$ and $I_k$ an interval in $\calT$, for each $k$. Let $Y(\calL) = (Y(L_1), \ldots, Y(L_k))^\T$, which we simply denote by $Y$, be the $K \times 1$ vector of outcomes observed at $\calL$. The key challenge of our model is that we do not observe the covariates that are assumed to explain the outcome at the same spatial-temporal resolution at which $Y$ is observed. Suppose $\tilde{\calL} = \{\lTilde_j: j = 1, \ldots, N\}$, where $\lTilde_j$ is of the form $(\tilde{s}_j, \ITilde_j)$ for each $j$, be the $N$ spatial-temporal coordinates in $\calS \times \Borel(\calT)$, at which we observe the covariate $X(\tilde{\calL}) = (X(\lTilde_1), \ldots, X(\lTilde_N))^\T$, which we simply write as the $N \times 1$ vector $X$. We further assume that $X(\tilde{\calL})$ are noisy measurements of a latent spatial-temporal process on $\calS \times \Borel(\calT)$ derived from a parent process $Z(\ell)$ on $\calD$ following the stochastic integral \eqref{eq:process-def}. This setup is particularly motivated by the asthma and ozone datasets described in Section~\ref{sec:data-description}, where the outcome is measured at the county level and aggregated annually, while the covariate is spatially referenced to points, but aggregated on a monthly scale. 

We build a hierarchical model by jointly modeling both the outcome and the covariate as conditionally dependent on a shared latent spatial-temporal process,
\begin{equation}\label{eq:hier-model}
\begin{split}
    Y(L_k) &= w(L_k)^{\T}\beta_1 + \beta_2 Z(L_k) + \epsilon_k,\quad
    \epsilon_k \ind \Norm \left(0, (|B_k| |I_k|)^{-1} \tau^2\right), k = 1, \ldots, K\;,\\
    (\beta_1^{\T}, \beta_2)^{\T} \given \tau^2 &\sim \Norm(\mu_{\beta}, \tau^2 V_{\beta}), \quad \tau^2 \sim \IG (a_{\tau}, b_{\tau})\;,\\
    X(\lTilde_j) &= Z(\lTilde_j) + e_j,\quad e_j \ind \Norm (0, |\ITilde_j|^{-1} \delta^2  \sigma^2), \quad j = 1, \ldots, N\;,\\
    Z(\ell) \given \gamma, \sigma^2, \phi &\sim \GP (\mu(\ell; \gamma), \sigma^2 C(\cdot, \cdot; \phi)), \quad \mu(\ell; \gamma) = \psi(\ell)^{\T}\gamma \;,\\
    \gamma \given \sigma^2 &\sim \Norm (\mu_\gamma, \sigma^2 V_\gamma), \quad \sigma^2 \sim \IG (a_\sigma, b_\sigma)\;,
\end{split}
\end{equation}
where $w(L_k) = (w_1(L_k), \ldots, w_p(L_k))^{\T}$ denotes a $p \times 1$ vector of predictors, $\beta_1$ is the corresponding $p \times 1$ vector of fixed effects, and, $\beta_2$ is a scalar and denotes the slope corresponding to the covariate $Z(L_k)$, which denotes the latent spatial-temporal process at $L_k$, for each $k$. Instead of $Z(L_K)$, we observe $X(\lTilde_j)$, which are noisy observations of $Z(\lTilde_j)$, which denotes the latent process at a different spatial-temporal resolution from $Z(L_k)$. Note that both $Z(L_k)$ for each $k$, and $Z(\lTilde_j)$ for each $j$, are completely unobserved, and their joint probability law is specified by the assumption of a Gaussian process ($\GP$) $Z(\ell)$ on $\calD$, following \eqref{eq:jointGP}. The mean function of the $\GP$ is characterized by a $r \times 1$ basis $\psi(\ell) = (\psi_1(\ell), \ldots, \psi_r(\ell))^{\T}$, and its corresponding $r \times 1$ slope vector $\gamma$. The purpose of $\psi(\ell)$ is to model a global trend of the spatial-temporal process which may capture an overall mean, seasonal variations, and so on. The term $\epsilon_k$ denotes white noise defined only at the coordinates in $\calL$, and captures the heteroskedasticity arising from inhomogeneity in the volumes of the spatial-temporal blocks at which the outcome is observed. Similarly, $e_j$ captures the independent measurement error for $X(\lTilde_j)$ for each $j$, and are defined only at $\tilde{\calL}$. It is crucial to note that, unlike $Z(\cdot)$, $\epsilon_j$ and $e_k$ are two different white noise processes and do not originate from a parent white noise process on $\calD$. Here, $\delta^2$ is the noise-to-spatial-temporal variance ratio. 

Let $W$ be the $K \times p$ matrix with $(k, u)$-th element $w_k(L_u)$ for $k = 1, \ldots, K$ and $u = 1, \ldots, p$. Collect the regression coefficients into the $(p+1) \times 1$ vector $\beta = (\beta_1^{\T}, \beta_2)^{\T}$. Given $\tau^2$, we place a multivariate Gaussian prior $\Norm (\mu_\beta, \tau^2 V_\beta)$ on $\beta$, and subsequently place an inverse-gamma prior $\IG (a_\tau, b_\tau)$ on $\tau^2$. Similarly, we also place a multivariate Gaussian prior $\Norm (\mu_\gamma, \sigma^2 V_\gamma)$ on $\gamma$ conditional on $\sigma^2$, and then place an inverse-gamma prior $\IG (a_\sigma, b_\sigma)$ on $\sigma^2$.
Define $K \times 1$ vector $Z_L = (Z(L_1), \ldots, Z(L_K))^{\T}$, and $N \times 1$ vector $Z_{\lTilde} = (Z(\lTilde_1), \ldots, Z(\lTilde_N))^{\T}$. 
Hence, we write the joint distribution of the data, latent process and the model parameters $p(Y, X, Z_L, Z_{\lTilde}, \beta, \gamma, \tau^2, \sigma^2)$ as
\begin{equation}\label{eq:joint-dist}
\begin{split}
    &\Norm (Y \given W\beta_1 + \beta_2 Z_L, \tau^2 D_L) \times \Norm (\beta \given \mu_\beta, \tau^2 V_\beta) \times \IG(\tau^2 \given a_\tau, b_\tau)\\
    &\quad \times \Norm(X \given Z_{\lTilde}, \delta^2 \sigma^2 D_{\lTilde}) \times \Norm (Z_{\lTilde} \given \tilde{\Psi}\gamma, \sigma^2 C_{\lTilde}(\phi)) \times \Norm (Z_L \given \mu_{L \given \lTilde}, \sigma^2 C_{L\given \lTilde})\\
    &\quad \quad \times \Norm (\gamma \given \mu_\gamma, \sigma^2 V_\gamma) \times \IG(\sigma^2 \given a_\sigma, b_\sigma) \;,
\end{split}
\end{equation}
where $D_L$ is $K \times K$ diagonal matrix with the $k$th diagonal element $(|B_k||I_k|)^{-1}$ for each $k$, and $D_{\lTilde}$ is $N \times N$ diagonal matrix with the $j$th diagonal element $|\ITilde_j|^{-1}$ for each $j$. Following \eqref{eq:jointGP}, we have
\begin{equation}\label{eq:conditional}
\begin{split}
    \mu_{L \given \lTilde} &= \bar{\Psi} \gamma + C_{\lTilde, L}(\phi)^{\T} C_{\lTilde}^{-1}(\phi)\left(Z_{\lTilde} - \tilde{\Psi} \gamma \right)\quad \mbox{and}\quad C_{L\given \lTilde} = C_L(\phi) - C_{\lTilde, L}(\phi)^{\T} C_{\lTilde}(\phi)^{-1} C_{\lTilde, L}(\phi)\;,
\end{split}
\end{equation}
where $N \times N$ matrix $C_{\lTilde}(\phi)$ and $N \times K$ matrix $C_{\lTilde, L}(\phi)$ are as defined in \eqref{eq:covs}. Furthermore, the $N \times r$ basis matrix $\tilde{\Psi}$ is known and has $(j, v)$-th element $\tilde{\psi}_v(\lTilde_j) = |\ITilde_j|^{-1} \int_{\ITilde_j} \psi_v(s, t) dt$ for $v = 1, \ldots, r$ and $j = 1, \ldots, N$. Similarly, the $K \times r$ basis matrix $\bar{\Psi}$ is known, with $(k, v)$-th element $\bar{\psi}_v(L_k) = (|L_k|)^{-1} \int_{L_k} \psi(\ell) d\ell$ for $k = 1, \ldots, K$ and $v = 1, \ldots, r$. Equation~\ref{eq:joint-dist} factorizes the joint distribution in a way that reflects the hierarchical structure of the data-generative model, as given by \eqref{eq:hier-model}.

\subsection{Model assumptions}\label{subsec:assumptions}
We pursue analytically accessible posterior distribution of all model parameters in \eqref{eq:hier-model} conditional on $\phi$ and $\delta^2$. The auxiliary model hyperparameters $\mu_\beta, V_\beta$, $\mu_\gamma, V_\gamma$, $a_\tau, b_\tau$, $a_\sigma$ and $b_\sigma$ are assumed fixed. Following \eqref{eq:joint-dist}, since $\beta$ and $\tau^2$ are conditionally independent of all model components, given $(Y, Z_L)$, we factorize the posterior distribution as
\begin{equation}\label{eq:posterior1}
    p(Z_L, Z_{\lTilde}, \beta, \gamma, \tau^2,\sigma^2 \given Y, X, \phi, \delta^2) = p(\beta, \sigma^2 \given Z_L, Y) \times p(Z_L, Z_{\lTilde}, \gamma, \sigma^2 \given Y, X, \phi, \delta^2)
\end{equation}
We make one additional assumption on $p(Z_L, Z_{\lTilde}, \gamma, \sigma^2 \given Y, X, \phi, \delta^2)$ as follows.
\begin{assumption}\label{assumption1}
    Under fixed values of $\phi$ and $\delta^2$, the latent spatial-temporal processes $(Z_L, Z_{\lTilde})$ is conditionally independent of the outcome $Y$, given noisy measurements of the covariate $X$, i.e.,
    $$p(Z_L, Z_{\lTilde}, \gamma, \sigma^2 \given Y, X, \phi, \delta^2) = p(Z_L, Z_{\lTilde}, \gamma, \sigma^2 \given X, \phi, \delta^2).$$
\end{assumption}
This implies that, instead of estimating a Bayesian full probability model, we assume that the latent spatial-temporal processes $(Z_L, Z_{\lTilde})$ are a priori dependent, but once $X$ is observed, $Y$ provides no additional information about them. In the context of our analysis, Assumption~\ref{assumption1} appears quite natural since it dictates that the estimation of annual county-level ozone concentrations only depends on $X$, i.e., the monthly ozone levels measured by the monitoring stations, and does not depend on asthma ED visit rates $Y$. In the literature of Bayesian inference for complex hierarchical models, this is a familiar approach, especially when there are multiple data sources that provide information about different parameters in the model \citep[see, for e.g.,][]{Bayarri2009, jacob2017}. This is commonly known as \emph{modularization}, as Assumption~\ref{assumption1} modulates the flow of information from observed data to the latent process. Moreover, our proposed framework can also be viewed as a \emph{cut model} \citep{Plummer2014}, since it mimics a cut in the directed acyclic graph representing the hierarchical model, separating the graph into two components by logically preventing the ``feedback'' from one part of the model to the other during inference. We illustrate this in Figure~\ref{fig:DAG}.
\begin{figure}[t]
    \centering
    \includegraphics[width=0.75\linewidth]{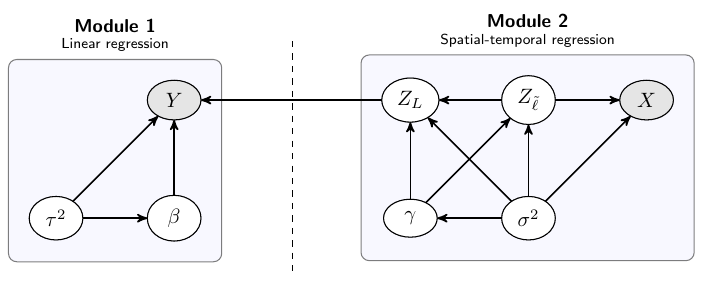}
    \caption{Directed acyclic graph (DAG) illustrating the conditional dependence structure of the hierarchical model \eqref{eq:hier-model}. Nodes shaded in gray correspond to data that are observed; all other nodes represent unobserved (latent or unknown) quantities. The vertical dashed line denotes a ``cut'' in the DAG, indicating restricted flow of information from $Y$ to Module~2 during posterior inference.}
    \label{fig:DAG}
\end{figure}
The separate components of the model are often called \emph{modules}. In this context, the hierarchical model \eqref{eq:hier-model} is made up of two modules -- a linear regression module (Module~1), and a spatial-temporal regression module (Module~2). In a full probability model, the information from the outcome $Y$ ``feeds back'' through the graph in Figure~\ref{fig:DAG} to influence the posterior distribution of the latent spatial-temporal processes $(Z_L, Z_{\lTilde})$. Instead, both $Z_L$ and $Z_{\lTilde}$ are estimated using the auxiliary data $X$. Cut models are also attractive from a computational viewpoint, as they significantly simplify sampling from the posterior distribution. 

\subsection{Posterior distribution}\label{subsec:posterior}
Following Assumption~\ref{assumption1}, conditional on $\phi$ and $\delta^2$, we derive that the posterior distribution corresponding to the hierarchical model in \eqref{eq:hier-model} as
\begin{equation}\label{eq:posterior-final}
\begin{split}
    & \Norm (\beta \given M_{\beta} m_{\beta}, \tau^2 M_{\beta}) \times \IG(\tau^2 \given a_{\tau}^{\ast}, b_{\tau}^{\ast}) \times \Norm (Z_L \given \mu_{L \given \lTilde}, \sigma^2 C_{L\given \lTilde})\\
    & \quad \times \Norm (Z_{\lTilde} \given M_{z} m_{z}, \sigma^2 M_{z}) \times \Norm (\gamma \given M_{\gamma} m_{\gamma}, \sigma^2 M_{\gamma}) \times \IG(\sigma^2 \given a_{\sigma}^{\ast}, b_{\sigma}^{\ast})\;,
\end{split}
\end{equation}
where $a_{\tau}^{\ast} = a_{\tau} + K/2$, $b_{\tau}^{\ast} = b_{\tau} + (Y^{\T} D_{L}^{-1} Y + \mu_{\beta}^{\T} V_{\beta}^{-1} \mu_{\beta} - m_{\beta}^{\T} M_{\beta} m_{\beta}) / 2$, and $a_{\sigma}^{\ast} = a_{\sigma} + N/2$, $b_{\sigma}^{\ast} = b_{\sigma} + (X^{\T} V_{X}^{-1} X + \mu_{\gamma}^{\T} V_{\gamma}^{-1} \mu_{\gamma} - m_{\gamma}^{\T} M_{\gamma} m_{\gamma}) / 2$, with $V_{X} = C_{\lTilde}(\phi) + \delta^2 D_{\lTilde}$, and
\begin{equation}\label{eq:postparams}
\begin{split}
    M_{\beta}^{-1} &= \tilde{W}^{\T} D_{L}^{-1} \tilde{W} + V_{\beta}^{-1}, \quad m_{\beta} = \tilde{W}^{\T} D_L^{-1} Y + V_{\beta}^{-1} \mu_{\beta}\;,\\
    M_{\gamma}^{-1} &= \tilde{\Psi}^{\T} V_{X}^{-1} \tilde{\Psi} + V_{\gamma}^{-1}, \quad m_{\gamma} = \tilde{\Psi}^{\T} V_{X}^{-1} X + V_{\gamma}^{-1} \mu_{\gamma}\;,\\
    M_{z}^{-1} &= C_{\lTilde}(\phi)^{-1} + (1/\delta^2) D_{\lTilde}^{-1}, \quad 
    m_z = (X - \tilde{\Psi}\gamma) / \delta^2\;, \\
\end{split}
\end{equation}
where $\tilde{W} = [W, Z_L]$ being the $K \times (p+1)$ matrix obtained by augmenting $W$ by the vector $Z_L$ on the right. The first two terms in \eqref{eq:posterior-final} denote the conditional posterior distributions $p(\beta \given \tau^2, Z_L, Y)$ and $p(\tau^2 \given Z_L, Y)$, respectively. The third term denote the posterior predictive distribution $p(Z_L \given Z_{\lTilde}, \gamma, \sigma^2, \phi, \delta^2)$, with $\mu_{L \given \lTilde}$ and $C_{L \given \lTilde}$ as defined in \eqref{eq:conditional}. The fourth through sixth terms represent the conditional posterior distributions $p(Z_{\lTilde} \given \gamma, \sigma^2, X, \phi, \delta^2)$, $p(\gamma \given \sigma^2, X, \phi, \delta^2)$, and the marginal posterior distribution $p(\sigma^2 \given X, \phi, \delta^2)$, respectively. This factorization follows from the fact that $Z_L$ is conditionally independent of the data $X$, given $Z_{\lTilde}$, $\gamma$, $\sigma^2$ and $\phi$. It is important to remark that, the analytical tractability of the posterior distribution arises from the assumption that $\phi$ and $\delta^2$ are fixed and due to the condition specified in Assumption~\ref{assumption1}.

We use composition sampling to draw samples from the posterior distribution in \eqref{eq:posterior-final}. We elaborate the steps in Algorithm~\ref{algo:sampling}. Here, $\chol(\cdot)$ refers to the lower-triangular Cholesky factor of a square matrix.%
\begin{algorithm}
\caption{Posterior Sampling}
\label{algo:sampling}

\KwIn{Data: $(Y, X, \tilde{\calL}, \calL)$; candidate $(\phi, \delta^2)$, and hyperparameters.}
\KwOut{Posterior samples of $\{\sigma^{2(b)}, \gamma^{(b)}, Z_{\lTilde}^{(b)}, Z_L^{(b)}, \tau^{2(b)}, \beta^{(b)}\}_{b=1}^B$}

Compute $\chol(V_X)$ given $\phi, \delta^2$\tcp*[r]{$\BigO(N^3)$ flops}

Compute $\chol(M_z)$ using $\chol(V_X)$, and $M_z = \delta^2 D_{\lTilde} V_X^{-1} C_{\lTilde}(\phi)$\tcp*[r]{$\BigO(N^3)$ flops}

Compute $\chol(C_{\lTilde}(\phi))$, then $\chol(C_{L \given \lTilde})$ using \eqref{eq:conditional}\tcp*[r]{$\BigO(KN^2 + K^3)$ flops}

Compute $M_\gamma$, $m_\gamma$, then $a_{\sigma}^{\ast}$, $b_{\sigma}^{\ast}$\tcp*[r]{$\BigO(rN^2 + r^3)$ flops}

\For{$b \gets 1$ \KwTo $B$}{
    Sample $\sigma^{2(b)} \sim \IG(a_{\sigma}^{\ast}, b_{\sigma}^{\ast})$\;

    Sample $\gamma^{(b)} \sim \Norm(M_{\gamma} m_{\gamma}, \sigma^{2(b)} M_{\gamma})$\;

    Compute $m_{z}^{(b)} = (X - \tilde{\Psi} \gamma^{(b)}) / \delta^2$\tcp*[r]{$\BigO(N)$ flops}

    Sample $Z_{\lTilde}^{(b)} \sim \Norm(M_z m_{z}^{(b)}, \sigma^{2(b)} M_z)$ using $\chol(M_z)$\tcp*[r]{$\BigO(N^2)$ flops}

    Compute $\mu_{L \given \lTilde}^{(b)}$ using $Z_{\lTilde}^{(b)}$ and $\gamma^{(b)}$ in \eqref{eq:conditional}\tcp*[r]{$\BigO(KN^2)$ flops}

    Sample $Z_L^{(b)} \sim \Norm (\mu_{L \given \lTilde}^{(b)}, \sigma^{2(b)} C_{L \given \lTilde} )$\tcp*[r]{$\BigO(N)$ flops}

    Set $\tilde{W}^{(b)} = [W, Z_L^{(b)}]$, compute $M_{\beta}^{(b)}$, $m_{\beta}^{(b)}$, $a_{\tau}^{\ast(b)}$, $b_{\tau}^{\ast(b)}$ using \eqref{eq:postparams}\tcp*[r]{$\BigO(K^3)$ flops}

    Sample $\tau^{2(b)} \sim \IG(a_{\tau}^{\ast(b)}, b_{\tau}^{\ast(b)})$\;

    Sample $\beta^{(b)} \sim \Norm(M_{\beta}^{(b)} m_{\beta}^{(b)}, \tau^{2(b)} M_{\beta}^{(b)})$\;
}
\Return $B$ samples $\{\sigma^{2(b)}, \gamma^{(b)}, Z_{\lTilde}^{(b)}, Z_{L}^{(b)}, \tau^{2(b)}, \beta^{(b)} \}_{b = 1}^{B}$ from the posterior distribution \eqref{eq:posterior-final}.
\end{algorithm}
The steps in Algorithm~\ref{algo:sampling}, generate $B$ samples $\{\sigma^{2(b)}, \gamma^{(b)}, Z_{\lTilde}^{(b)}, Z_{L}^{(b)}, \tau^{2(b)}, \beta^{(b)} \}_{b = 1}^{B}$ from the posterior distribution. The sampling algorithm is dominated by Cholesky decompositions of $N \times N$ matrices $V_X$, $M_z$, and $C_{\lTilde}(\phi)$, which accumulates $\BigO(N^3)$ floating-point operations (flops). The matrix $M_z$ is calculated efficiently using the identity $M_z = \delta^2 D_{\lTilde} V_{X}^{-1} C_{\lTilde}(\phi)$, where $V_X^{-1} C_{\lTilde}(\phi)$ is calculated using triangular solves of columns of $C_{\lTilde}(\phi)$ with respect to the Cholesky factor already calculated $\chol(V_X)$. Placing a prior on $\phi$ and the nugget $\theta = \delta^2 \sigma^2$ requires iterative algorithms such as MCMC to sample from the posterior distribution, which entails repeated Cholesky decompositions of these $N \times N$ matrices. Moreover, weak identifiability of $\phi$, $\sigma^2$ and $\theta$ impedes convergence of the random walk Metropolis steps, which contributes to delayed execution times. Thus, even for moderately sized datasets ($N \sim 10^3$), the computation becomes too onerous for practical use and alternative strategies like low-rank models are used \citep{spbayes_large}. In this context, predictive stacking presents itself as an effective alternative by enhancing the practicality of full Gaussian process models for moderately large datasets.
\begin{remark}\label{remark1}
Given observations at $\tilde{\calL}$, let $\calL = \{\ell_i = (s_i, t_i): i = 1, \ldots, N^{\ast} \}$ be a collection of $N^{\ast}$ space-time coordinates in $\mathcal{D}$, where we wish to predict latent spatial-temporal process. Define $N^{\ast} \times 1$ vector $Z_{\ell} = (Z(\ell_1), \ldots, Z(\ell_{N^{\ast}}))^{\T}$. Then, posterior predictive inference at $\calL$ follows from 
\begin{equation}\label{eq:spt_pointpred}
    p(Z_{\ell} \given X, \phi, \delta^2) = \int p(Z_{\ell} \given Z_{\lTilde}, \gamma, \sigma^2, \phi, \delta^2)\, p(Z_{\lTilde}, \gamma, \sigma^2 \given X, \phi, \delta^2)\, dZ_{\lTilde}\, d\gamma \, d\sigma^2 \;,
\end{equation}
where the conditional density $p(Z_{\ell} \given Z_{\lTilde}, \gamma, \sigma^2, \phi, \delta^2) = \Norm (Z_{\ell} \given \mu_{\ell \given \lTilde}, \sigma^2 C_{\ell \given \lTilde})$, with
\begin{equation}\label{eq:conditional2}
\begin{split}
    \mu_{\ell \given \lTilde} = {\Psi} \gamma + C_{\ell, \lTilde}(\phi) C_{\lTilde}^{-1}(\phi)\left(Z_{\lTilde} - \tilde{\Psi} \gamma \right)\;,\quad
    C_{\ell\given \lTilde} = C_{\ell}(\phi) - C_{\ell, \lTilde}(\phi) C_{\lTilde}(\phi)^{-1} C_{\ell, \lTilde}(\phi)^{\T}\;,
\end{split}    
\end{equation}
where $(\mu_\ell(\gamma))_i = \mu(\ell_i; \gamma)$, and $N^{\ast} \times N$ matrix $C_{\ell, \lTilde}(\phi)$, and $N^{\ast} \times N^{\ast}$ matrix $C_{\ell}(\phi)$ are defined as 
\begin{equation*}
    (C_\ell(\phi))_{i,i'} = C(\ell_i, \ell_{i'}; \phi), \quad (C_{\ell, \lTilde}(\phi))_{i, j} = |\ITilde_j|^{-1} \int_{\ITilde_j} C((s_i, t_i), (\tilde{s}_j, t); \phi) dt \;,
\end{equation*}
for $i, i' = 1, \ldots, N^{\ast}$ and $j = 1, \ldots, N$. We sample from \eqref{eq:spt_pointpred} by first drawing $\{Z_{\lTilde}^{(b)}, \gamma^{(b)}, \sigma^{2(b)}\}$ from $p(Z_{\lTilde}, \gamma, \sigma^2 \given X, \phi, \delta^2)$ using Algorithm~\ref{algo:sampling}. Then, for each drawn value of $\{Z_{\lTilde}^{(b)}, \gamma^{(b)}, \sigma^{2(b)}\}$, we sample $Z_{\ell}$ from $\Norm(Z_{\ell} \given \mu_{\ell \given \lTilde}^{(b)}, \sigma^{2(b)} C_{\ell \given \lTilde})$, where $\mu_{\ell \given \lTilde}^{(b)}$ is obtained by substituting $Z_{\lTilde}$ and $\gamma$ by $Z_{\lTilde}^{(b)}$ and $\gamma^{(b)}$, respectively, in \eqref{eq:conditional2}. Repeating this for $b = 1, \ldots, B$ yields samples $\{Z_{\ell}^{(b)}: b = 1, \ldots, B\}$ from \eqref{eq:spt_pointpred}.
\end{remark}

\subsection{Spatial-temporal correlation function}\label{subsec:sptcorr}
As can be seen from \eqref{eq:covs}, evaluation of the elements of $C_{\lTilde}(\phi)$ involves numerical integrations over the temporal domain. Popular methods for univariate numerical integration (e.g., trapezoidal/Simpson's rule, quadrature, Monte Carlo) for a smooth function incurs $\BigO(M)$ operations, where $M$ denotes the number of knots/points at which the integrand is evaluated. Hence, evaluation of all elements of $C_{\lTilde}(\phi)$ would require $\BigO(MN^2)$ flops, which can be computationally expensive for moderately large $N$. Therefore, we explore spatial-temporal correlation functions that offer improved tractability, especially over the temporal domain.

In particular, we assume a separable spatial-temporal correlation function \citep{mardiagoodall1993} of the form $C(\ell, \ell'; \phi) = C_s(s, s'; \phi_s, \nu) \cdot C_t(t, t'; \phi_t)$, where $C_s(s, s'; \phi_s, \nu)$ and $C_t(t, t'; \phi_t)$ denote the isotropic Mat\'ern and the exponential correlation functions, respectively, given by
\begin{equation}\label{eq:corrfn}
\begin{split}
    C_s(s, s'; \phi_s, \nu) &= \frac{\left(\phi_s \lVert s - s' \rVert \right)^{\nu}}{2^{\nu - 1} \Gamma(\nu)} K_{\nu} \left( \phi_s \lVert s - s' \rVert \right) \;,\\
    C_t(t, t'; \phi_t) &= \exp (-\phi_t |t-t'|)\;,
\end{split}
\end{equation}
where $\lVert s - s' \rVert$ is the Euclidean distance between $s, s' \in \calS$. The function $\Gamma (\cdot)$ denotes the gamma function, and $K_{\nu}$ is the modified Bessel function of the second kind of order $\nu$ which may be fractional \citep[Chapter 10]{abramowitzstegun}. Hence, in this case, the process parameter $\phi = (\phi_s, \nu, \phi_t)$ is a 3-dimensional parameter. The assumption of this multiplicative form conveniently separates space and time in calculation of covariance matrices $C_{\lTilde}(\phi)$, $C_{L}(\phi)$, and the cross-covariance matrix $C_{\lTilde, L}(\phi)$. For example, the elements of $C_{\lTilde}(\phi)$ can be rewritten as
\begin{equation}\label{eq:ClTilde}
    (C_{\lTilde}(\phi))_{j, j'} = (|\ITilde_j||\ITilde_{j'}|)^{-1} C_{s}(\tilde{s}_{j}, \tilde{s}_{j'}; \phi_s, \nu) \int_{\ITilde_{j'}} \int_{\ITilde_j} C_t(t, t'; \phi_t) dt dt'\;.
\end{equation}
We make an additional assumption that facilitates closed form expression for the integral in \eqref{eq:ClTilde}.
\begin{assumption}\label{assumption2}
    The temporal blocks $\ITilde_j$ for $j = 1, \ldots, N$ and $I_k$ for $k = 1, \ldots, K$, are of the form $\ITilde_j = (\tilde{a}_j, \tilde{b}_j)$, and $I_k = (a_k, b_k)$ where $\tilde{a}_j < \tilde{b}_j$ for each $j$ and $a_k < b_k$ for each $k$.
\end{assumption}
The distribution theory in Section~\ref{subsec:spt-process} applies to any Borel-measurable subsets $\ITilde_j$ and $I_k$. In practice, however, it is natural to assume that observations are aggregated over a single interval.
\begin{proposition}\label{prop:timeCov}
    For any $a<b$ and $c<d$, suppose $\tilde{C}_t(a, b, c, d; \phi) = \int_{c}^{d} \int_{a}^{b} C_t(t, t'; \phi_t) dt dt'$, then
    \begin{itemize}
        \item[(a)] for non-overlapping $(a, b)$ and $(c, d)$, with $a < b \leq c < d$,
        \[
        \tilde{C}_t(a, b, c, d; \phi_t) = \frac{1}{\phi_t^2} \left[ F(a, d) + F(b, c) - F(a, c) - F(b, d) \right]\;,
        \]
        \item[(b)] if intervals $(a, b)$ and $(c, d)$ overlap, with $a \leq c < b \leq d$,
        \[
        \tilde{C}_t(a, b, c, d; \phi_t) = \frac{1}{\phi_t^2} \left[ 2 \phi_t (b-c) + F(a, d) + F(c, b) - F(a, c) - F(b, d) \right]\;,
        \]
        \item[(c)] if $(a, b)$ is nested within $(c, d)$, i.e. either $c \leq a < b < d$ or $c < a < b \leq d$,
        \[
        \tilde{C}_t(a, b, c, d; \phi_t) = \frac{1}{\phi_t^2} \left[ 2\phi_t (b-a) + F(a, d) + F(c, b) - F(c, a) - F(b, d) \right]\;,
        \]
    \end{itemize}
    where $F(a_1, a_2) = F(a_1, a_2; \phi_t) = \exp \left( -\phi_t (a_2 - a_1) \right)$ for any $a_1, a_2 \in \calT$.
\end{proposition}
\begin{proof}
    See Appendix~\ref{app:technical} for details.
\end{proof}
We benefit from the customary Assumption~\ref{assumption2} that the integral of $C_{t}(\cdot, \cdot; \phi_t)$ as appearing in \eqref{eq:ClTilde} admits closed-form expressions, as detailed in Proposition~\ref{prop:timeCov}. This facilitates a rapid evaluation of elements of the spatial-temporal covariance matrix $C_{\lTilde}(\phi)$ without resorting to approximations using numerical methods. However, Proposition~\ref{prop:timeCov} does not extend to similar results for integrations over irregular spatial blocks (e.g., counties of California), and hence numerical methods are essential for computing elements of $C_{L}(\phi)$ and $C_{\lTilde, L}(\phi)$. In this context, assuming the separable spatial-temporal covariance function \eqref{eq:corrfn} provides additional advantage, as we are able to write the integral as a product of a spatial and a temporal component. Following \eqref{eq:covs} and Assumption~\ref{assumption2}, we have
\begin{equation}\label{eq:separableIntegral}
\begin{split}
    (C_{\lTilde, L} (\phi))_{j, k} &= \frac{|B_k|^{-1}}{(\tilde{b}_j - \tilde{a}_j)(b_k - a_k)} \tilde{C}_t(\tilde{a}_j, \tilde{b}_j, a_k, b_k) \int_{B_k} C_s (\tilde{s}_j, s) ds\;,\\
    (C_L(\phi))_{k, k'} &= \frac{(|B_k||B_{k'}|)^{-1}}{(b_k - a_k)(b_{k'} - a_{k'})} \tilde{C}_t (a_k, b_k, a_{k'}, b_{k'}) \int_{B_{k'}} \int_{B_k} C_s(s, s') ds ds'\;,
\end{split}
\end{equation}
for each $j = 1, \ldots, N$ and $k, k' = 1, \ldots, K$. The role of Proposition~\ref{prop:timeCov} becomes clear from \eqref{eq:separableIntegral}, as it helps to simplify \eqref{eq:covs} by reducing the necessity of any numerical integration in the time domain. From a computational perspective, due to Proposition~\ref{prop:timeCov}, evaluation of $C_{\lTilde}(\phi)$ requires absolutely no numerical integration, $C_{\lTilde, L}(\phi)$ requires $NK$ numerical integrals and $C_L(\phi)$ requires $K^2$ numerical integrals. In practice, $K$ is generally much smaller compared to $N$. For example, California has only 58 counties, and we include observations for 5 racial groups over 8 years from 2015 through 2022, which amounts to $K = 1510$ after removal of missing records. On the other hand, we have recorded ozone levels at $N = 15,725$ spatial-temporal coordinates.

\section{Predictive stacking}\label{sec:stacking}
\subsection{Choice of candidate models}\label{subsec:candidate}
To implement stacking, we first fix the values of the hyperparameters of the auxiliary model. In practice, we assume $\mu_\beta = 0_{p+1}$, $V_\beta = \delta_\beta^2 I_{p+1}$ for a sufficiently large $\delta_\beta$ to specify a weakly informative prior for the fixed effects $\beta$. Here, $0_{p+1}$ denotes the zero vector of length $p+1$. Similarly, we assume $\mu_\gamma = 0_r$ and $V_\gamma = \delta^2_{\gamma} I_r$ for the Gaussian prior on $\gamma$. For $\tau^2$ and $\sigma^2$, we choose the shape parameters $a_{\tau} = a_{\sigma} = 2$ and the scale parameters $b_{\tau} = b_{\sigma} = 0.1$. 

For $\phi = (\phi_s, \phi_t, \nu)$ and $\delta^2$, we choose grids of candidate values given by $G_{\phi_s}$, $G_{\phi_t}$, $G_\nu$ and $G_{\delta^2}$. The grid $G_{\delta^2}$ is chosen based on the values of the nugget and partial sill, estimated from an empirical semivariogram. \cite{zhang2024stacking} provides further details on how the quantiles of a Beta distribution specified by the estimated values of the nugget and the partial sill can be used to provide useful information for selecting $G_{\delta^2}$. On the other hand, the candidate values of $\phi_s$ and $\phi_t$ are chosen so that the ``effective range'' (distance at which the correlation drops below 5\%) corresponding to the candidate values is between 20\% and 70\% of the maximum distance between the space-time coordinates \citep[see, Chapter~2][]{banerjee_spatial}. We choose $G_{\nu}$ to include some customary values of the Mat\'ern smoothness parameter that are often used in the literature of spatial analysis. For example, a possible choice is $G_{\nu} = \{0.5, 1.0, 1.5\}$.

\subsection{Stacking algorithm}\label{subsec:stackingalgo}
While we follow the general strategy of stacking predictive densities as proposed in \cite{yao2018using}, our development is distinct in that we adapt the approach to our more complex model structure, which differs significantly from the standard spatial-temporal modeling frameworks considered in the previous literature \citep{zhang2024stacking, pan2025stacking}. In this aspect, Assumption~\ref{assumption1} plays a key role, as it restricts the inference for the latent spatial-temporal process to Module~2, within which we seek to find an optimal way to combine the inference conditional on candidate values of $\phi$ and $\delta^2$. We elaborate below.

Let $\calM = \{M_1, \ldots, M_G\}$ denote a collection of candidate models $G$, where $M_g$ corresponds to fixed values of the parameters $(\phi_g, \delta^2_g)$, for $g = 1, \ldots, G$. Predictive stacking finds a probability distribution $\tilde{p}$ in the class $\mathcal{C} = \{ \sum_{g = 1}^G \alpha_g p(\cdot \given X, M_g) : \sum_{g=1}^G \alpha_g = 1, \alpha_g \geq 0\}$, such that the Kullback-Leibler ($\mathrm{KL}$) divergence between $\tilde{p}(\cdot \given X)$ and $p_t(\cdot \given X)$ is minimized, where $p_t$ denotes the posterior predictive distribution under the true data-generating model. Here, $p(\cdot \given X, M_g)$ denotes the posterior predictive distribution under the candidate model $M_g$, for each $g$. We define the stacking weights $\alpha = (\alpha_1, \ldots, \alpha_G)$ as the solution to the optimization problem
\begin{equation}\label{eq:stack_optim}
\begin{split}
    \max_{\alpha}\quad &\frac{1}{N} \sum_{j = 1}^{N} \log \sum_{g = 1}^{G} \alpha_g \, p\left( X(\lTilde_j) \biggiven X_{-j}, M_g \right)\quad
    \text{subject to}\quad \alpha^{\T} 1_{G} = 1, \ \alpha \in [0, 1]^{G}\;,
\end{split}
\end{equation}
where $X_{-j}$ denotes the data $X$ with the $j$th observation removed, and $p(X(\lTilde_j) \given X_{-j}, M_g)$ denotes the leave-one-out predictive density corresponding to the $j$th observation. This follows from a result that establishes that minimizing $\mathrm{KL} \left(\tilde{p} (\cdot \given X) , p_{t}(\cdot \given X) \right)$ under the constraint $\tilde{p} \in \mathcal{C}$ is asymptotically equivalent to the optimization problem in \eqref{eq:stack_optim} \citep[see, ][]{le2017bayes, clyde2013bayesian}. The optimization task in \eqref{eq:stack_optim} falls into the class of convex problems and can be formulated and solved using suitable modeling tools and solvers. Assumption~\ref{assumption1} implies that the optimal stacking weights can be computed solely based on Module~2, that is, using only $X$. Posterior inference for quantities of interest subsequently proceeds from the ``stacked posterior'',
\begin{equation}
    \tilde{p}(\cdot \given X, Y) = \sum_{g = 1}^{G} \hat{\alpha}_g \, p(\cdot \given X, Y, M_g)\;,
\end{equation}
where $\hat{w}_g$ denotes optimal stacking weights obtained by solving the optimization task in \eqref{eq:stack_optim}. 

An important prerequisite for evaluating the objective function of interest in \eqref{eq:stack_optim} is the computation of the leave-one-out predictive densities. Under the hierarchical model \eqref{eq:hier-model}, the leave-one-out predictive densities admit a closed form, given by
\begin{equation}\label{eq:exactLOOPD}
    p\left( X(\lTilde_j) \biggiven X_{-j}, M_g \right) = \tdist_{2 a_{\sigma, j}^{\ast}} \left( X(\lTilde_j) \biggiven \mu_{g,j \given -j} (\lTilde_j), \left( b_{\sigma, j}^{\ast} / a_{\sigma, j}^{\ast} \right) \sigma_{g, j\given -j} \right)\;,
\end{equation}
where $\tdist_{\rho}(x; m, v^2)$ denotes the location-scale $t$-density with degrees of freedom $\rho$, location $m$ and scale $v$, evaluated at $x$, and the parameters $\mu_{g, j \given -j} (\lTilde_j)$ and $\sigma_{g, j\given -j}$ are given by
\begin{equation*}
\begin{split}
    \mu_{g, j \given -j} (\lTilde_j) &= R_{g, j}^{\T} V_{g,X_{-j}} X_{-j} + H_{g, j}^{\T} M_{\gamma, g, j} m_{\gamma, g, j}\;,\\
    \sigma_{g, j\given -j} &= V_{g, X_j \given X_{-j}} + H_{g, j}^{\T} M_{\gamma, g, j} H_{g, j}\;,
\end{split}
\end{equation*}
where $R_{g, j}$ is the $(N-1) \times 1$ spatial-temporal cross-correlation matrix between $\tilde{\calL} \setminus {\lTilde_j}$ and $\{\lTilde_j\}$ under $M_g$, and is given by the $j$th column of $C_{\lTilde, -j}(\phi_g)$ which denotes the matrix $C_{\lTilde}(\phi_g)$ with both its $j$th row and column removed. Moreover, the matrix $V_{g, X_{-j}}$ is obtained by removing the $j$th row and column of $V_{g, X} = C_{\lTilde}(\phi_g) + \delta^2_g D_{\lTilde}$. In addition, $M_{\gamma, g, j}^{-1} = \tilde{\Psi}_{-j}^{\T} V_{g, X_{-j}} \tilde{\Psi}_{-j} + V_{\gamma}^{-1}$, and $m_{\gamma, g, j} = \tilde{\Psi}_{-j} V_{g, X_{-j}} X_{-j} + V_{\gamma}^{-1} \mu_{\gamma}$ where the $(N-1) \times r$ matrix $\tilde{\Psi}_{-j}$ is obtained by deleting the $j$th row of $\tilde{\Psi}$. The $r \times 1$ vector $H_{g, j} = \tilde{\psi}(\lTilde_j) - \tilde{\Psi}_{-j}^{\T} V_{g, X_{-j}}^{-1} R_{g, j}$, and the scalar $V_{g, X_j \given X_{-j}} = (V_{g, X})_{j,j} - R_{g, j}^{\T} V_{g, X_{-j}} R_{g, j}$ with $(V_{g, X})_{j,j}$ denoting the $j$th diagonal element of $V_{g, X}$. Furthermore, $a_{\sigma, g}^{\ast} = a_{\sigma} + (N-1)/2$, and $b_{\sigma, g}^{\ast} = b_\sigma + (X_{-j}^{\T} V_{g, X_{-j}} X_{-j} + \mu_{\gamma}^{\T} V_{\gamma}^{-1} \mu_{\gamma} - m_{\gamma, g, j}^{\T} M_{\gamma, g, j} m_{\gamma, g, j})/2$. 

Evaluation of \eqref{eq:exactLOOPD} is dominated by the Cholesky decomposition of the $(N-1)\times(N-1)$ matrix $V_{g, X_{-j}}$, which requires $\BigO(N^3)$, for each $j$. Hence, a naive approach to find the leave-one-out predictive densities under each $M_g$ results in $\sim \BigO(N^4)$ flops, which is impractical. We mitigate this issue by reusing the Cholesky factor of $V_{g, X}$ which has already been computed once while fitting the model $M_g$ \citep{kimcox2002_CV}. Subsequently, we compute the Cholesky factor of $V_{g, X_{-j}}$ for each $j$ using an efficient rank-one update algorithm \citep{krauseIgel2015}, which ultimately accumulates to $\BigO(N^3)$ flops, thereby delivering a significant speedup over the naive approach.

Although rank-one updates are faster than the naive method, they are computationally expensive; hence, an alternative approach is desirable. Here, importance weighting is an attractive option to approximate leave-one-out predictive densities \citep[see, for e.g.,][]{gelfand92_IS, PSIS_vehtari24}. Suppose $\{Z_{\lTilde, g}^{(b)}, \sigma_{g}^{2(b)}\}_{b = 1}^{B}$ denotes $B$ draws from the posterior distribution $p(Z_{\lTilde}, \sigma^2 \given X, M_g)$, then for each $j$, we approximate the leave-one-out predictive densities by the weighted mean
\begin{equation}\label{eq:psis}
    p\left( X(\lTilde_j) \biggiven X_{-j}, M_g \right) \approx \frac{1}{\sum_{b = 1}^{B} r_{j, g}^{b}} \sum_{b = 1}^{B} r_{j, g}^{b} \Norm \left( X(\lTilde_j) \biggiven Z_{\lTilde, g, j}^{(b)}, \sigma_{g}^{2(b)} \right) \;,
\end{equation}
where $Z_{\lTilde, g, j}^{(b)}$ denotes the $j$th element of $Z_{\lTilde, g}^{(b)}$, and $r_{j, g}^{b}$ is the important ratio defined as $1/r_{j, g}^{b} = \Norm(X(\lTilde_j) \given Z_{\lTilde, g, j}^{(b)}, \sigma_{g}^{2(b)})$. The weights $r_{j, g}^{b}$ tend to have a high or infinite variance, introducing instability in the computation \eqref{eq:psis}. To address these issues, \cite{LOOCV_vehtari17} proposes stabilizing the weights by fitting a generalized Pareto distribution to the tail of the weight distribution using the empirical Bayes estimation algorithm proposed in \cite{zhang2009_pareto}. Thus, no additional model fitting is necessary to calculate leave-one-out predictive densities. Compared to $\BigO(N^3)$ for the Cholesky factor update algorithm, the computational cost of this approximate method is $O(N \log N)$ and therefore is much faster. 
Figure~\ref{fig:exact-psis} in the Appendix compares exact and approximate leave-one-out predictive densities computed using the closed-form expression \eqref{eq:exactLOOPD} and Pareto smoothed importance sampling, respectively, for a spatial regression model. Both methods deliver practically indistinguishable results but differ in computational cost; exact computation requires roughly 3~minutes, whereas importance sampling needs only 0.1~second. We implement this using the \textsf{R} package \texttt{loo} \citep{loo_pkg}. 

\section{Simulation}\label{sec:simulation}
\subsection{Simulated data}\label{sec:sim-data}
To address the analytical challenges posed by data that are simultaneously spatially and temporally misaligned, we designed a simulation study that closely replicates the structure and characteristics of the outcome (asthma) and the exposure (ozone) data, as discussed in Section~\ref{sec:data-description}. For the exposure data, we consider the unit square, i.e., $[0, 1]^2$, and the interval $[0, 12]$ as the spatial and temporal domains of interest. The interval $[0, 12]$ represents the duration of twelve months of a year. We simulate spatially-temporally correlated data at $n_s = 100$ locations sampled uniformly in $[0, 1]^2$, and at $n_t = 360$ equally spaced time points in $[0, 12]$, which amounts to spatial-temporal coordinates $\{\ell_1, \ldots, \ell_{n_{s}n_{t}}\}$ with $n_s n_t = 36,000$. The data is simulated following the model $X(\ell_i) = Z(\ell_i) + e'_i$ for $i = 1, \ldots, n_s n_t$, where $e'_i \ind \Norm(0, 1)$ denotes the measurement error for each $i$ and $Z(\ell) \sim \GP(\mu(\ell), C(\ell, \ell'; \phi))$. The covariance function $C(\ell, \ell'; \phi)$ is the same as \eqref{eq:corrfn} with $\phi_s = 4$, $\nu = 0.5$ and $\phi_t = 0.6$. We choose the mean function $\mu(\ell)$ to depend only on time and therefore write it as $\mu(t)$. We model $\mu(t)$ as a random draw from $\GP(5, 4\cdot R(t, t'; p, \lambda))$, where $R(t, t'; p, \lambda) = \exp(-2 \lambda^2 \sin^2\left(\pi |t - t'| / p\right))$ denotes a periodic covariance kernel with period $p = 7$ and decay $\lambda=0.1$. This helps introduce seasonal variations in the simulated data. Thus, the simulated data $X(\ell_i)$ for $i = 1, \ldots, n_s n_t$ mimic daily measurements of a variable at 100 monitoring sites over a year. From these, we compute monthly averages at each site, which we denote by $X(\lTilde_j)$ for $j = 1, \ldots, 12 n_s$. To introduce missingness, we randomly remove 10\% of the data from the monthly averages and denote the remaining data simply by $X$. We proceed to our simulation experiment with these point-referenced monthly averages as the data. See Figures~\ref{fig:sim-oz}A~and~\ref{fig:sim-oz}B for a visualization of the simulated data.
\begin{figure}[t]
    \centering
    \includegraphics[width=0.95\linewidth]{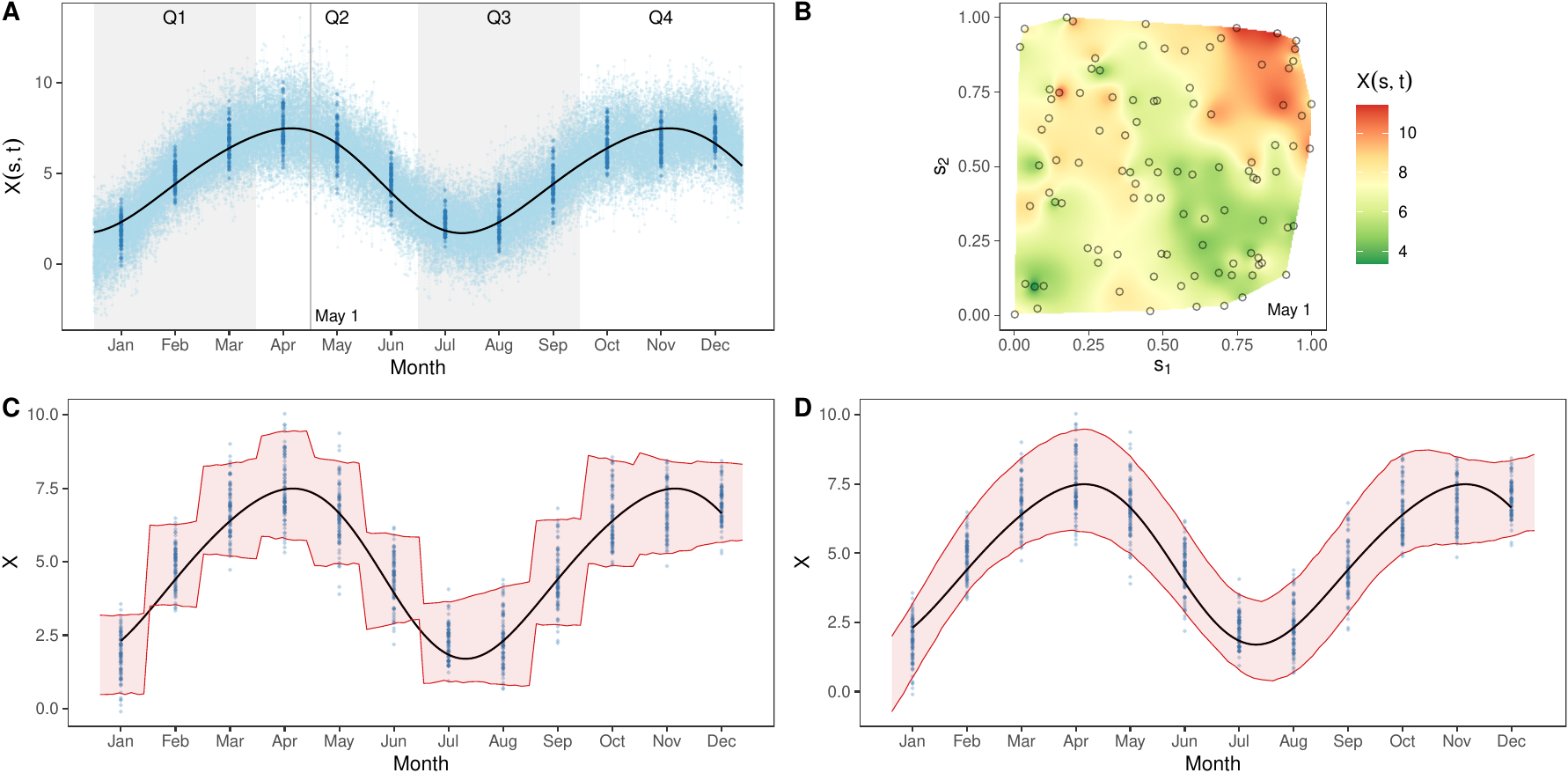}
    \caption{{\sffamily A}: Light blue dots represents simulated daily measurements at 100 locations over a year, whereas dark blue dots denote monthly averages at each site. {\sffamily B}: Interpolated spatial surface of a snapshot of the data (denoted by the vertical line in subfigure~{\sffamily A}) at May 1. {\sffamily C}, {\sffamily D}: 95\% credible interval of posterior predictive samples drawn from the stacked posterior $\tilde{p}(Z(\ell) \given X)$ with $\psi(\cdot)$ chosen as monthly factors and Fourier series as basis, respectively, to capture seasonal variations.}
    \label{fig:sim-oz}
\end{figure}
For subsequent analysis, we consider the grids of candidate values of the spatial decay parameter $G_{\phi_s} = \{2, 3, 5\}$, $G_{\phi_t} = \{0.3, 0.5, 1\}$, $G_{\nu} = \{0.5, 1, 1.5\}$, and $G_{\delta^2} = \{0.75, 1.5\}$. Hence, we stack on $3 \times 3 \times 3 \times 2 = 54$ models.

\subsection{Point-level spatial-temporal prediction}\label{sec:sim-pred1}
As a first step, we examine the stacked posterior predictive distribution $\tilde{p}(X(\ell) \given X)$, which corresponds to predicting the underlying process at a finer temporal resolution (e.g., daily) from coarser, aggregated observations (see Remark~\ref{remark1}). In this context, $\psi(\cdot)$ in the mean function $\psi(\cdot)$ of the latent spatial-temporal process plays a key role. Under the assumption that $\mu(\ell) = \psi(t)^{\T} \gamma$ depends only on time, we study the two alternative specifications of the mean function
\begin{equation}\label{eq:basis}
\begin{split}
    \mu^{(1)}(t) = \sum_{v = 1}^r \psi^{(1)}_v(t)\gamma_v &= \gamma_1 + \sum_{v = 2}^r  \gamma_{v} \mathsf{1}(t \in \mathrm{month}_v)\;,\\
    \mu^{(2)}(t) = \sum_{v = 1}^r \psi^{(2)}_v(t)\gamma_v &= \gamma_1 + \sum_{v = 1}^{\lfloor r/2 \rfloor} \gamma_{2v} \sin(2\pi t/p_v) + \gamma_{2v+1} \cos(2\pi t/p_v)\;,
\end{split}
\end{equation}
where $\mathsf{1}(\cdot)$ denotes an indicator function, $\mathrm{month}_{v}$ denotes the interval $(v-1, v)$ for $v = 2, \ldots, 12$. For example, $v = 2$ corresponds to the month February. In this case, $r = 12$ and $\psi^{(1)}(t)$ corresponds to a simple monthly indicator basis function, which corresponds to monthly factors under monthly aggregated data. On the other hand, $\mu^{(2)}(t)$ consists of smooth periodic basis functions, such as sine and cosine terms with different periodicity. For our analysis, we choose $r = 9$ with $p_v = 3+v$ for $v = 1, \ldots, 4$. Figures~\ref{fig:sim-oz}C~and~\ref{fig:sim-oz}D illustrate the behavior of the posterior predictive distribution under the two choices of $\psi(\cdot)$. We notice that the smooth basis functions lead to superior reconstruction of the latent process compared to discontinuous monthly indicator functions, due to their ability to borrow strength across adjacent time points and capture underlying seasonal trends more effectively.

\subsection{Block-level spatial-temporal prediction}\label{sec:sim-pred2}
Next, we study prediction of the latent process at spatial-temporal blocks. Based on monthly data available at the 100 locations, our aim is to predict quarterly aggregated data in a collection of spatial blocks, instead of points. Here, a quarter corresponds to consecutive three-month periods, as shown by alternating shaded regions in Fig.~\ref{fig:sim-oz}A. We divide the spatial domain, the unit square $[0, 1]^2$, for the four quarters into 40, 50, 30, and 60 irregular spatial blocks, respectively, obtained by a Voronoi tessellation based on the same number of randomly sampled points, using the \textsf{R} package \texttt{deldir} \citep{deldir_pkg}. 
Given the simulated data, we seek the posterior predictive distribution of the aggregated latent process quarterly in these spatial blocks throughout the quarters, based on the stacked posterior $\tilde{p}(Z_L \given X)$. We assume that the mean function $\mu(\cdot)$ is characterized by the monthly indicator basis function $\psi^{(1)}$, as defined in \eqref{eq:basis}. Figure~\ref{fig:sim-COS} compares the quarterly aggregated true detrended spatial-temporal surface $Z(\ell) - \mu(\ell)$, with which the data were simulated, and its corresponding posterior median in the spatial blocks.
\begin{figure}[t]
    \centering
    \includegraphics[width=0.95\linewidth]{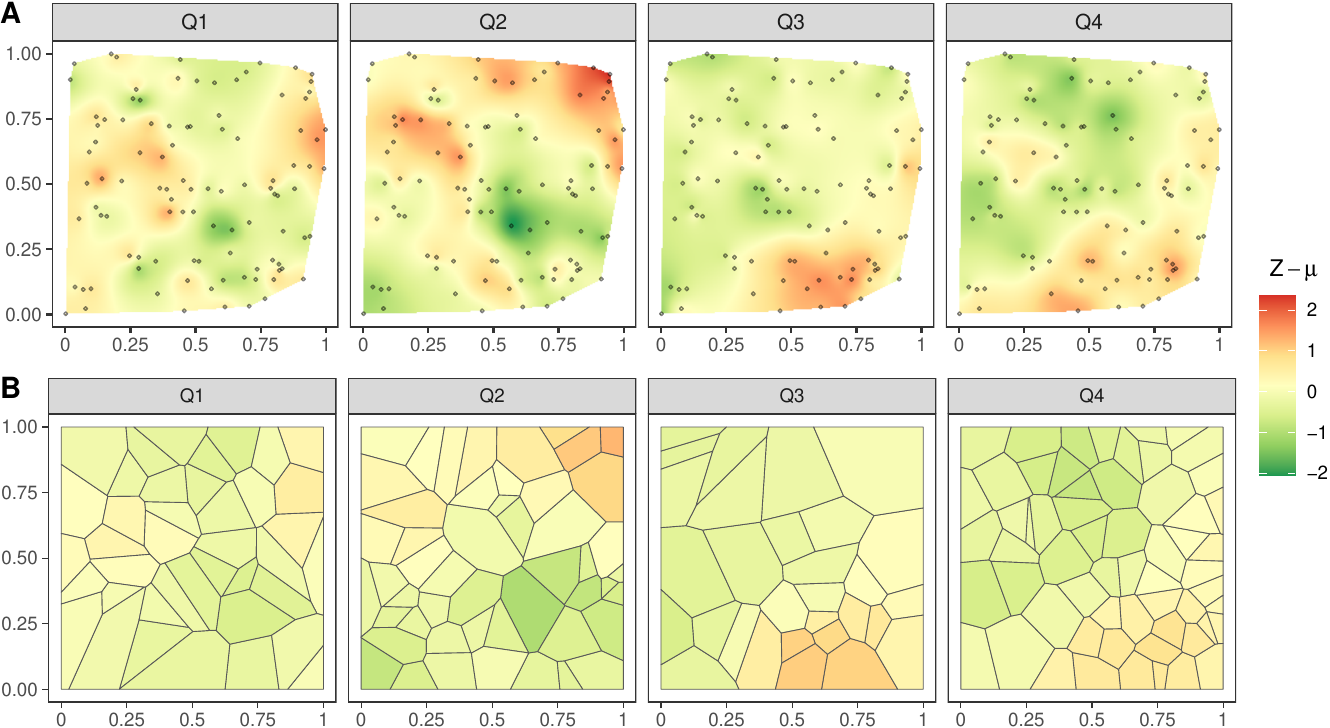}
    \caption{\textsf{A}: Interpolated spatial surfaces of the quarterly averaged de-trended true spatial-temporal process which simulated the data; \textsf{B}: Median of the stacked posterior predictive distribution $\tilde{p}(Z_L - \mu_L \given X)$ at the target spatial-temporal blocks.}
    \label{fig:sim-COS}
\end{figure}
We notice that the block-level posterior distributions closely reflect the true spatial-temporal pattern, indicating that our proposed model accurately captures the underlying process. This further demonstrates the flexibility of our proposed hierarchical model \eqref{eq:hier-model}, as it accommodates complicated survey designs with time-varying spatial blocks.

\subsection{Model comparison}\label{sec:model-comp}
We evaluate and compare the predictive performance of our proposed model against some alternative approaches. For alternative models, we regress the result $Y(L_k)$ for each $k$, as defined in \eqref{eq:hier-model}, on the predictors $w(L_k)$ and $X(L_k)$. Here, $X(L_k)$ denotes the value of the covariate aggregated in the spatial-temporal block $L_k$ and is distinct from the latent process $Z(L_k)$. We use available off-the-shelf spatial interpolation tools to estimate $X(L_k)$. More specifically, we consider two competing approaches - multilevel B-spline approximation (\textsf{MBA}), implemented using the \textsf{R} package \texttt{MBA} \citep{mba_pkg}, and, spatial kriging (\textsf{kriging}) from the \textsf{R} package \texttt{geoR} \citep{geoR_pkg}. Unlike our proposed method, \textsf{MBA} and \textsf{kriging} proceed by first obtaining point estimates in a fine grid in the spatial domain and then calculating averages in each spatial block, completely ignoring the uncertainty surrounding the estimation procedure. In addition, none of the alternative methods account for temporal correlation in the observed data, and temporal aggregation at the quarterly level is achieved by simply averaging over monthly observations. 

The predictive accuracy for each model is measured using the widely applicable information criterion, WAIC \citep{waic2010, gelman2013}. We evaluated WAIC for each model on synthetic data, which are simulated following the steps described in Section~\ref{sec:sim-data} with the exception that we now also sample the latent process corresponding to the quarterly averages in spatial blocks, using \eqref{eq:jointGP}, where integrated covariance kernels are approximated through Monte Carlo integration with 500 within-polygon samples. Subsequently, we sample an outcome $Y$ following the model described in \eqref{eq:hier-model} with $p = 2$, and the covariate $w(\cdot)$ containing an intercept and a predictor sampled from the standard normal distribution, $\beta_1 = (5, 1)^{\T}$, $\beta_2 = -1$ and $\tau^2 = 5$. From Table~\ref{tab:waic}, we notice that all methods deliver similar predictive performance, with our proposed method slightly better.
\begin{table}[t]
    \centering
    \begin{tabular}{cccccc}
        \toprule
        \multirow{2}{*}{Method} & \multirow{2}{*}{\begin{tabular}[c]{@{}c@{}}Model-based\\ UQ \end{tabular}} & \multirow{2}{*}{\begin{tabular}[c]{@{}c@{}}Temporal\\ Dependence \end{tabular}} & \multirow{2}{*}{WAIC} & \multirow{2}{*}{$n_{\text{grid}}$} & \multirow{2}{*}{\begin{tabular}[c]{@{}c@{}}Run Time\\ (in secs) \end{tabular}}\\ \\ \midrule
        \multirow{3}{*}{\textsf{MBA}} & \multirow{3}{*}{\xmark} & \multirow{3}{*}{\xmark} & 1714.80 & 200 & 281\\ 
        & & & 1714.81 & 100 & 71\\ 
        & & & 1714.85 & 50 & 17\\ \cmidrule{1-6}
        \multirow{3}{*}{\textsf{kriging}} & \multirow{3}{*}{\xmark} & \multirow{3}{*}{\xmark} & 1714.75 & 200 & 326\\ 
        & & & 1714.73 & 100 & 81\\ 
        & & & 1714.75 & 50 & 20\\ \cmidrule{1-6}
        \textsf{Stacking} & \cmark & \cmark & 1714.53 & -- & 250\\
         \bottomrule
    \end{tabular}
    \caption{Comparison of predictive performance of our proposed method based on stacking of predictive densities with the alternative approaches. Here, $n_{\text{grid}}$ refers to the number of points along each axis at which spatial interpolation is performed to estimate block-level observations; UQ: Uncertainty quantification.}
    \label{tab:waic}
\end{table}
Moreover, the execution times for both \textsf{MBA} and \textsf{kriging} depends on the grid resolution. Also, contrary to common assumption, increasing the grid resolution has little impact on predictive accuracy, probably because finer resolution spatial interpolation contributes little when they are aggregated to obtain block-level estimates. Executing the stacking algorithm corresponds to fitting 54 candidate models in parallel across 6 cores. Hence, in addition to offering competitive predictive performance with reasonable run-time, a crucial advantage of our framework over other methods is the ability to deliver fully model-based uncertainty quantification for the latent spatial-temporal process at any arbitrary location and time.

\section{Data analysis}\label{sec:data-analysis}
We study the effects of ozone concentration on asthma-related ED visits by first estimating the spatial-temporal regression module based on available monthly ozone concentration data at 15,725 space-time coordinates. We use the periodic Fourier basis function, as given by $\psi^{(2)}(\cdot)$ in \eqref{eq:basis}. Guided by visual inspection of the temporal trends in ozone measurements (see Figure~\ref{fig:oz-panel} in the Appendix), we consider periods of 5, 6, 7, and 12. We consider the separable spatial-temporal covariance kernel \eqref{eq:corrfn}. For candidate values of $\phi$, we choose $G_{\phi_s} = \{0.3, 0.5, 1\}$, $G_{\nu} = \{0.5, 1, 1.5\}$, and $G_{\delta^2} = \{1.5, 2\}$. Hence, we stack on the 54 models obtained by the Cartesian product of each of these grids. Once we obtain optimal stacking weights based on 1000 posterior samples of $\{Z_{\lTilde}, \gamma, \sigma^2\}$, we perform posterior predictive inference for ozone concentrations at 20,000 unique space-time coordinates, comprising 100 randomly sampled locations within the convex hull of monitoring sites at 200 time points spanning 2015–2022. Figure~\ref{fig:oz-panel-pred} displays the 95\% posterior predictive credible intervals for the latent spatial-temporal process at these space-time locations, overlaid on the observed ozone measurements. 
\begin{figure}[t]
    \centering
    \includegraphics[width=0.95\linewidth]{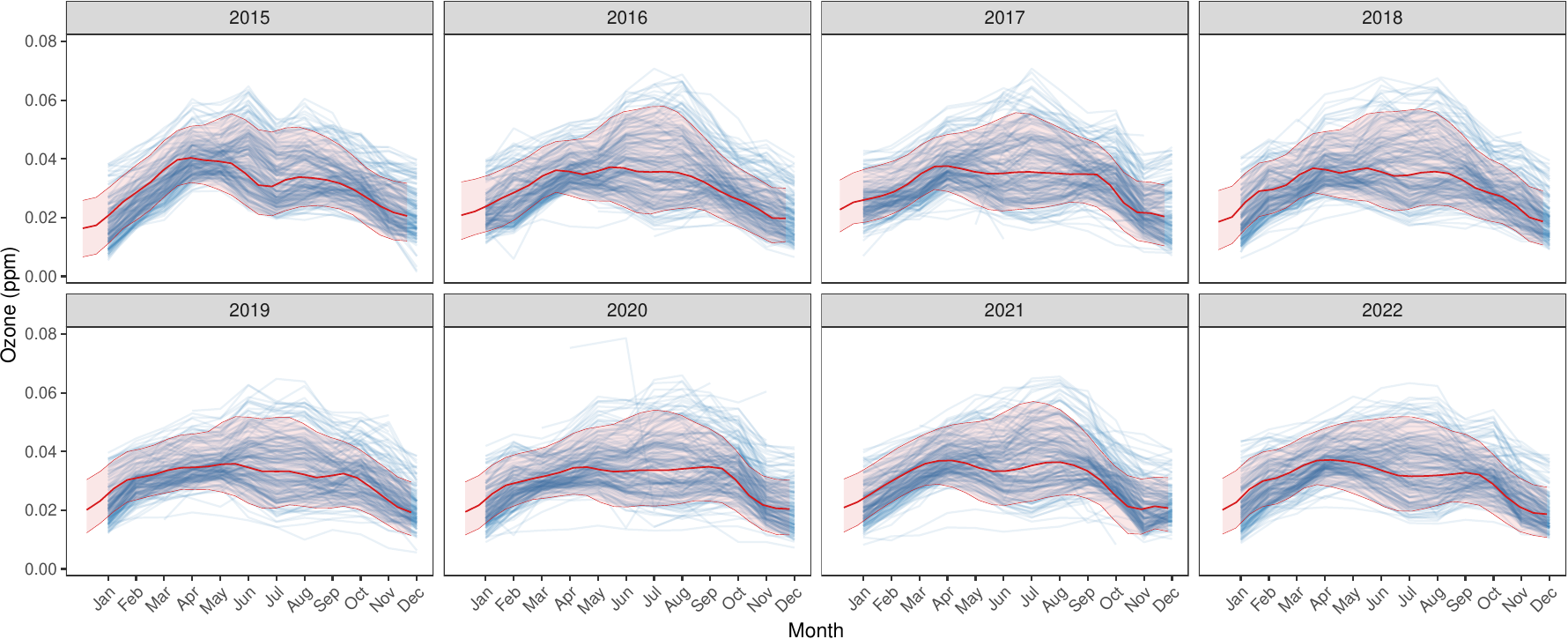}
    \caption{95\% credible intervals (in \emph{red}) of ozone concentration predictions obtained from the stacked posterior $\tilde{p}(Z(\ell) \given X)$, using a periodic Fourier basis mean.}
    \label{fig:oz-panel-pred}
\end{figure}
The uncertainty band is estimated from the stacked posterior $\tilde{p}(Z(\ell) \given X)$. We notice that the temporal trend is captured reasonably well.

\begin{figure}[t]
    \centering
    \includegraphics[width=0.95\linewidth]{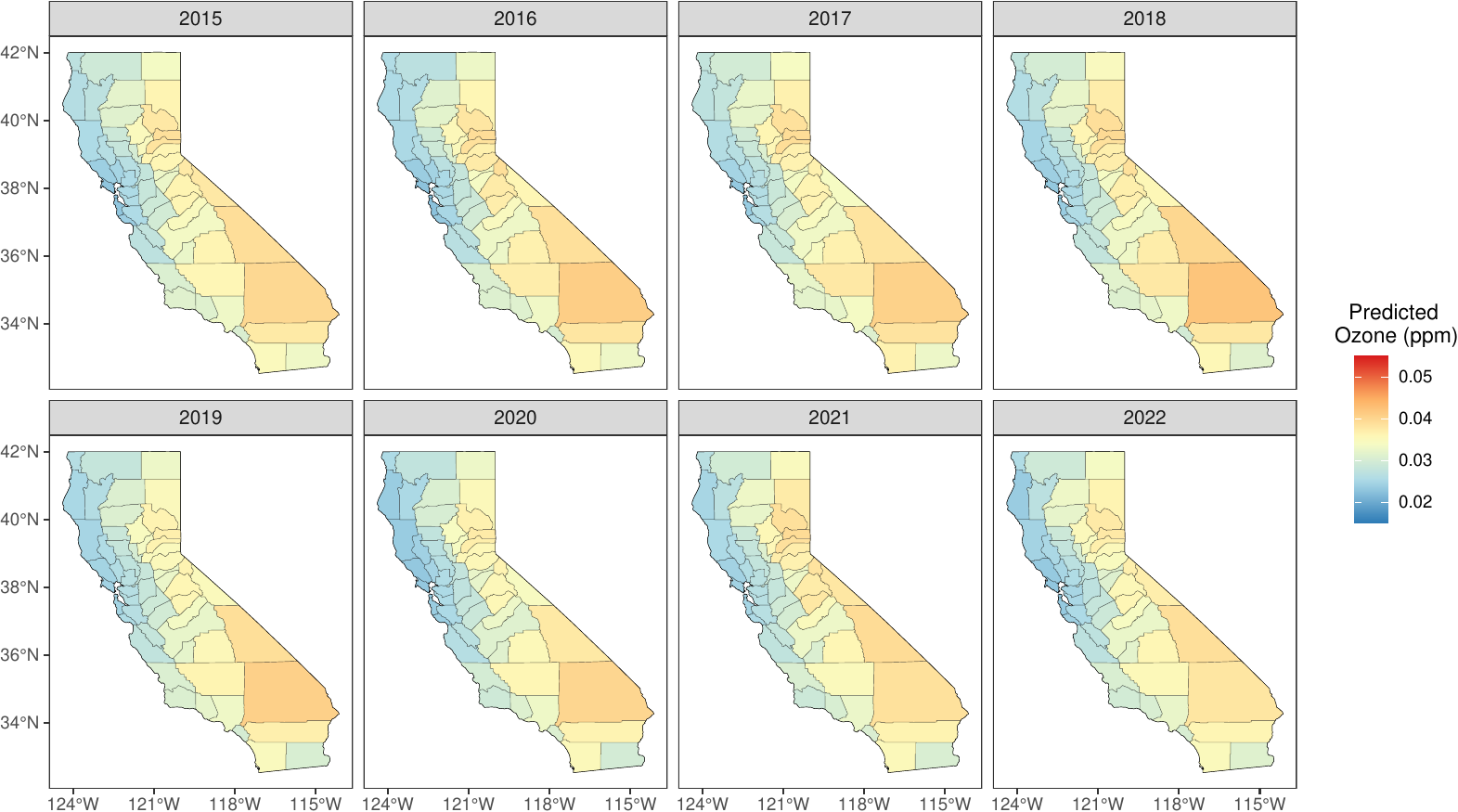}
    \caption{Posterior median of annual ozone concentration predictions at counties of California, obtained from monthly point-referenced observations.}
    \label{fig:oz-county}
\end{figure}
Next, we obtain annual county-level ozone concentration predictions based on the stacked posterior $\tilde{p}(Z(L) \given X)$. Figure~\ref{fig:oz-county} shows the posterior median of annual ozone concentration predictions in each county. The spatial patterns closely resemble those in Figure~\ref{fig:oz-surface}, which shows interpolated surfaces from aggregated monthly monitoring data. However, the pronounced ozone hotspots (in the counties of Riverside, San Bernadino, Inyo, Sierra, Plumas, El Dorado) and the low-concentration regions (e.g., Humboldt, Mendocino, Sonoma, Marin, San Francisco, San Mateo, Santa Cruz) in Figure~\ref{fig:oz-surface} are somewhat less pronounced in their corresponding counties in Figure~\ref{fig:oz-county}. As anticipated, spatial averaging dampens local fluctuations, resulting in smoother estimates that may blur sharp variations in the observed data; nevertheless, the overall pattern of higher inland regions remains largely preserved. We use these posterior predictive samples of annual county-level ozone concentrations to estimate its effect on the asthma-related ED visits.

Subsequently, we estimate a linear regression module with the objective of studying the effect of ozone, race/ethnicity, and time, on age-adjusted asthma-related emergency department visit rates. For clarity, we reformulate the linear model in \eqref{eq:hier-model} using symbolic notation as follows. Suppose $Y_{ijt}$ denotes the age-adjusted asthma-related ED visit rate (per 10,000) for county $i$, race/ethnicity $j$ and year $t$, for $i = 1, \ldots 58$, $j = 1, \ldots, 5$ and $t = 1, \ldots, 8$. We consider the log-linear model
\begin{equation}\label{eq:lm}
    \log (Y_{ijt}) = \beta_0 + \beta_1 \cdot \textsf{Race}_{j} + \beta_2 \cdot \textsf{Ozone}_{it} + \beta_3 \cdot \textsf{Year}_t + \epsilon_{ijt}\;, \quad \epsilon_{ijt} \ind \Norm(0, |B_{i}|^{-1} \tau^2)
\end{equation}
where $\beta_0$ is the intercept, $\textsf{Race}_j$ is a $4 \times 1$ dummy-encoded vector representing the $j$th racial group, $\beta_1 = (\beta_{11}, \beta_{12}, \beta_{13}, \beta_{14})^{\T}$ is the corresponding $4 \times 1$ vector of regression coefficients, $\textsf{Ozone}_{it}$ is the estimated ozone concentration for county $i$ in year $t$, $\textsf{Year}_t = t$ for each $t$, $|B_i|$ is the area of county $i$, and $\epsilon_{ijt}$ denote independent measurement errors. We assume a hierarchical model surrounding \eqref{eq:lm} as given by \eqref{eq:hier-model}. We assume the Gaussian prior $\beta \given \sigma^2 \sim \Norm(0, \sigma^2 V_\beta)$ with $V_\beta = 10^3 I_7$, where $\beta = (\beta_0, \beta_1^{\T}, \beta_2, \beta_3)^{\T}$, and place an inverse gamma prior $\sigma^2 \sim \IG (0.01, 0.01)$. 

\begin{table}[t]
    \centering
    \begin{tabular}{llccl}
    \toprule
       \multirow{2}{*}{Parameter} &  \multirow{2}{*}{Effect} & \multirow{2}{*}{\begin{tabular}[c]{@{}c@{}}Posterior\\ median \end{tabular}} & \multirow{2}{*}{\begin{tabular}[c]{@{}c@{}} 95\% credible\\ interval \end{tabular}} & \multirow{2}{*}{Details} \\ \\
    \midrule
    $\beta_0$  & Intercept & 3.53 & (3.49, 3.58) & \\
    \cmidrule{1-5}
    & \emph{Reference group:} White & & &\\
    $\beta_{11}$ & American Indian/ Alaskan native & 0.24 & (0.16, 0.31) & \\
    $\beta_{12}$ & Asian/ Pacific Islander & -0.69 & (-0.76, -0.62) & \\
    $\beta_{13}$ & Black & 1.30 & (1.24, 1.36) & \\
    $\beta_{14}$ & Hispanic & 0.004 & (-0.05, 0.06) & \\
    \cmidrule{1-5}
    $\beta_2$ & Ozone (per 0.005 ppm) & -0.025 & (-0.05, 0.00) & \emph{baseline:} 0.03 ppm\\
    $\beta_3$ & Year & -0.10 & (-0.11, -0.09) & \emph{baseline:} 2015\\
    $\sigma^2$ & Error variance & 0.13 & (0.12, 0.14) & \\
    \bottomrule
    \end{tabular}
    \caption{Asthma and ozone data 2015-22: posterior summary of model parameters.}
    \label{tab:posterior_summary}
\end{table}

We present a comprehensive summary of the posterior distributions of the regression coefficients in Table~\ref{tab:posterior_summary}.
We find that all regression coefficients except $\beta_{14}$ have 95\% posterior credible intervals that exclude zero, suggesting strong evidence of meaningful differences in contributions specific to each racial group, except that the Hispanic and White groups do not differ significantly in their effects on asthma-related health emergencies, for a given year and ozone level. For better understanding, we examine the relative effect sizes corresponding to the racial groups of white ($e^{\beta_0}$), American Indian/ Alaskan native ($e^{\beta_0 + \beta_{11}}$), Asian/ Pacific Islander ($e^{\beta_0 + \beta_{12}}$), black ($e^{\beta_0 + \beta_{13}}$), and Hispanic ($e^{\beta_0 + \beta_{14}}$) racial groups, which reflects the expected rates of asthma-related ED visits for each racial group in the reference year 2015, when exposed to the baseline level of ozone of 0.03 ppm. Figure~\ref{fig:effects}A illustrates the posterior distributions of the relative effects sizes of each racial group.
\begin{figure}[t]
    \centering
    \includegraphics[width=0.9\linewidth]{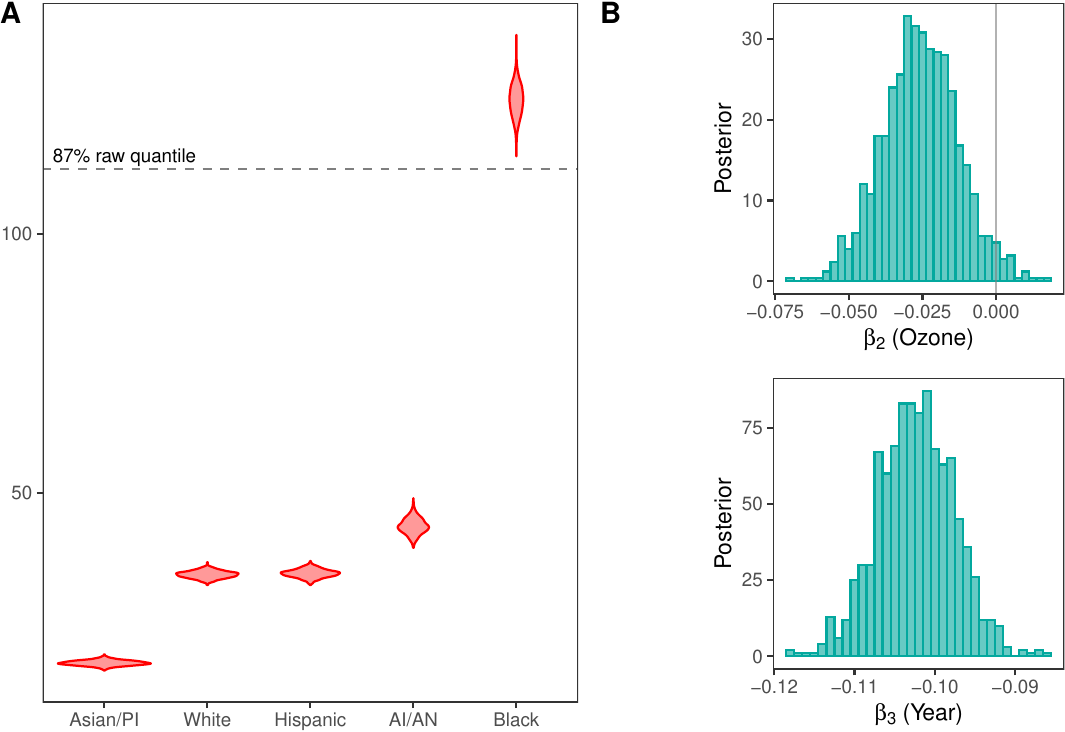}
    \caption{\textsf{A}: Posterior distributions of relative effect sizes at constant ozone concentration and year; AI/AN: American Indian/ Alaskan native, PI: Pacific Islanders. \textsf{B}: Posterior distributions of the coefficients corresponding to ozone and year.}
    \label{fig:effects}
\end{figure}
Our analysis reveals that the average asthma-related ED visit rates of the black group are approximately 3.6 times the average rates of the white (CI: 3.4, 3.8) and Hispanic (CI: 3.4, 3.9) groups, 2.9 (CI: 2.7, 3.1) times the average rate of the American Indian/Alaskan native group, and around 7.3 (CI: 6.8, 7.9) times the average rate of the Asian / Pacific Islander group, when all are exposed to the same levels of ozone concentration, during the year 2015. Here, the reported values represent posterior medians, and ``CI'' denotes the corresponding 95\% posterior credible intervals. Moreover, as seen from Figure~\ref{fig:effects}A, the relative effect size of the black group lies above the 87th percentile of the observed rates. These substantial differences highlight the persistent and disproportionate burden of asthma-related emergency department visits among the black population in California.

Figure~\ref{fig:effects}B displays the posterior distributions of the regression coefficients corresponding to ozone concentration and year. We observe a marginally negative association between ozone and asthma-related ED visit rates. This suggests a weak inverse relationship, with higher ozone levels associated with slightly lower asthma-related ED visit rates. Although the association is not statistically strongly supported, it is consistent with patterns reported in previous studies \citep{ZhuCarlinGelfand2003}. This also contrasts with the commonly held assumption that a higher concentration of ozone increases the risk of asthma-related symptoms. More specifically, we find that an increase of 0.005 ppm of ozone concentration leads to a drop in asthma-related ED visit rates by a factor of 0.97 (CI: 0.94, 1). This relationship suggests seasonal interactions of ozone \citep{quick2015jrssc}. Ozone levels also tend to be higher in rural inland areas (e.g., Central Valley), population density is lower, and ED utilization is lower due to access barriers, which does not necessarily imply a lower incidence of asthma. Urban coastal areas (e.g., San Francisco Bay area, Los Angeles) may have lower ozone but higher asthma rates due to other pollutants, higher population density and reporting, and different healthcare-seeking behaviors. Our analysis shows that each year is associated with a 10\% annual decrease in asthma-related ED visit rates, with the expected rate decreasing by a factor of 0.9 (CI: 0.89–0.91). This is consistent with the downward trend in Figure~\ref{fig:asthma-panel} of the Appendix. 

We fit the model on 15,725 space-time coordinates using a stacked posterior that combines 54 candidate models. The entire procedure took roughly 90 minutes, a considerable improvement over traditional full Bayesian MCMC approaches, where a single iteration could take over 15 minutes. Furthermore, our regression inference framework accounts for the uncertainty in the estimated ozone concentrations by integrating the samples from their posterior distribution. This ensures that both the latent process and its downstream effects are quantified in a fully probabilistic and computationally efficient manner. Moreover, the use of periodic Fourier basis functions for modeling temporal trends enables smooth interpolation and prediction at any time point within the study window. This flexibility would be difficult to achieve using conventional discrete-time models. In addition, a model relying on monthly basis functions would be unable to provide predictions for months with no available data, which is a likely scenario given the irregularity frequently encountered in environmental monitoring records.

\section{Discussion}\label{sec:discussion}
A key strength of our modular Bayesian framework is its capacity to jointly estimate ozone concentrations at arbitrary spatial and temporal resolutions, along with their association with an outcome of interest. Stacking of predictive densities enables us to obtain fully model-based uncertainty quantification for all model parameters by averaging over a collection of candidate models, each representing different specifications of process parameters. Future methodological directions may involve developing a multivariate areal time series model in place of the current linear regression framework, with the goal of capturing additional sources of variability and complex temporal-spatial dependencies inherent in the data. The data analysis presented in this article not only underscores existing racial disparities in health outcomes but also points to systemic inequities in environmental exposure, access to healthcare, and underlying social determinants of health. It emphasizes the need for targeted public health interventions and policies that address the structural drivers of asthma morbidity and improve health equity in racial and ethnic communities. Future data analysis may consider jointly estimating multiple exposure variables, such as $\text{NO}_2$, $\text{PM}_{2.5}$, and others, using multivariate spatial-temporal regression models to better account for potential correlations among pollutants and their combined effects on health outcomes.

\section*{Acknowledgement}
\noindent This work used computational and storage services of the Hoffman2 Shared Cluster provided by UCLA Office of Advanced Research Computing's Research Technology Group. The authors were supported by research grant R01ES027027 from the National Institute of Environmental Health Sciences (NIEHS), R01GM148761 from the National Institute of General Medical Science (NIGMS) and DMS2113778  from the Division of Mathematical Sciences of the National Science Foundation (NSF-DMS).

\section*{Data and code availability}
\noindent The asthma data is openly available at the \href{https://data.chhs.ca.gov/dataset/asthma-emergency-department-visit-rates}{CalHHS} website hosted by the California Department of Public Health (CDPH). The data on ozone measurements and geographic locations of the air quality monitoring stations are collected using the Air Quality and Meteorological Information System (\href{https://www.arb.ca.gov/aqmis2/aqdselect.php}{AQMIS}) database query tool of the California Air Resource Board (CARB). The code is provided as an \textsf{R} package, \texttt{spStackCOS}, available for download from the GitHub repository at \url{https://github.com/SPan-18/spStackCOS-dev}. Supplementary code to reproduce the results and findings presented in this article can be found at \url{https://github.com/SPan-18/AsthmaOzoneCA}.

\section*{Conflicts of interest}
\noindent All authors declare that they have no conflicts of interest.


\bibliographystyle{plainnat}
\bibliography{refs}

\appendix
\newpage
\begin{center}
\textbf{\large Appendix to ``Bayesian inference for spatial-temporal non-Gaussian data\\using predictive stacking''}
\end{center}
\setcounter{equation}{0}
\setcounter{figure}{0}
\setcounter{table}{0}
\setcounter{lemma}{0}
\setcounter{theorem}{0}
\setcounter{proposition}{0}
\makeatletter
\renewcommand{\theequation}{A\arabic{equation}}
\renewcommand{\thetheorem}{A\arabic{theorem}}
\renewcommand{\thelemma}{A\arabic{lemma}}
\renewcommand{\theproposition}{A\arabic{proposition}}
\renewcommand{\thefigure}{A\arabic{figure}}

\section{Descriptive overview of the dataset}
In this section, we present an exploratory data analysis of age-adjusted emergency department (ED) visit rates per 10,000 residents in California, across counties and racial groups. Figure~\ref{fig:asthma-race} display annual county-level ED visit rates for each racial group throughout the study period 2015-22, revealing clear spatial and racial patterns. Notably, the black population consistently reports higher ED visit rates across the study period, with elevated rates concentrated in the coastal areas of northern California. 
\begin{figure}[t]
    \centering
    \includegraphics[width=0.75\linewidth]{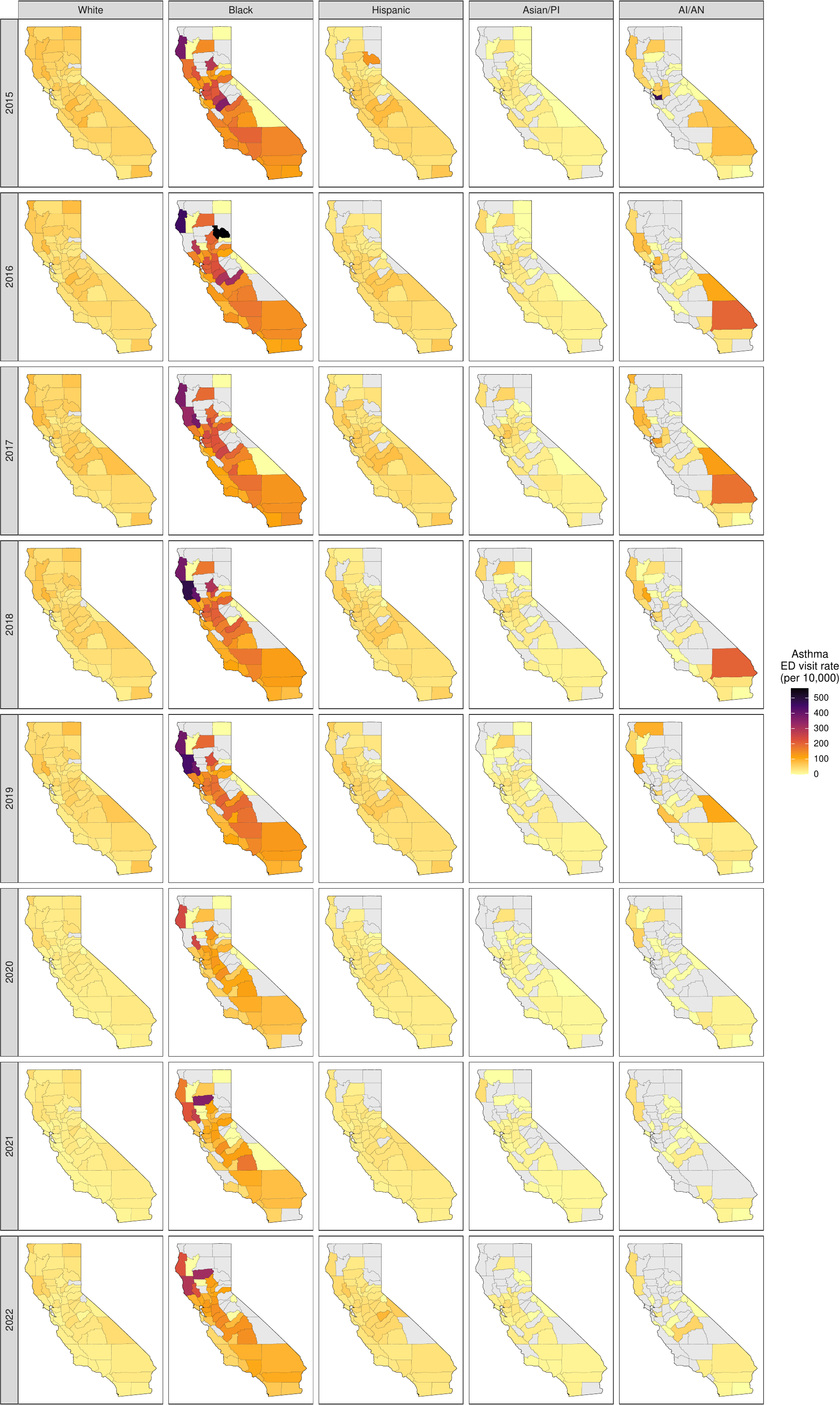}
    \caption{Asthma ED visit rates by race 2015-22}
    \label{fig:asthma-race}
\end{figure}
In terms of data completeness, the white racial group has 2.1\% missing data and the black group has 10.8\% missing data, while the Asian/Pacific Islander and American Indian/Alaska Native groups exhibit high levels of missingness, with 34.1\% and 59.3\% missing respectively. On the other hand, at the county level, Los Angeles, Riverside, Sacramento and San Diego counties have complete ED visit data. The counties with the most missing data are Lassen, Siskiyou, Calaveras, Del Norte, Plumas, Mariposa, Modoc, Glenn, Nevada, and Trinity, each exhibiting approximately 50–60\% missingness. These counties are predominantly rural and located in Northern California and the Sierra Nevada region. These areas are characterized by mountainous terrain and lower population density, which may contribute to challenges in healthcare access and data reporting completeness. These patterns underscore the importance of considering both spatial and demographic dimensions in analyzing the data.

Next, we plot the ED visit rates over time to examine temporal trends within each racial group. As shown in Figure~\ref{fig:asthma-panel}, most groups exhibit a gradual decline in rates over the study period, suggesting potential improvements in underlying health conditions, access to preventive care, or changes in reporting and healthcare utilization. A strong declining trend is particularly evident for the white and the Asian/Pacific Islander groups while other groups show more variability. These patterns provide important context for interpreting cross-sectional differences and support the inclusion of temporal components in our modeling framework.
\begin{figure}[t]
    \centering
    \includegraphics[width=0.95\linewidth]{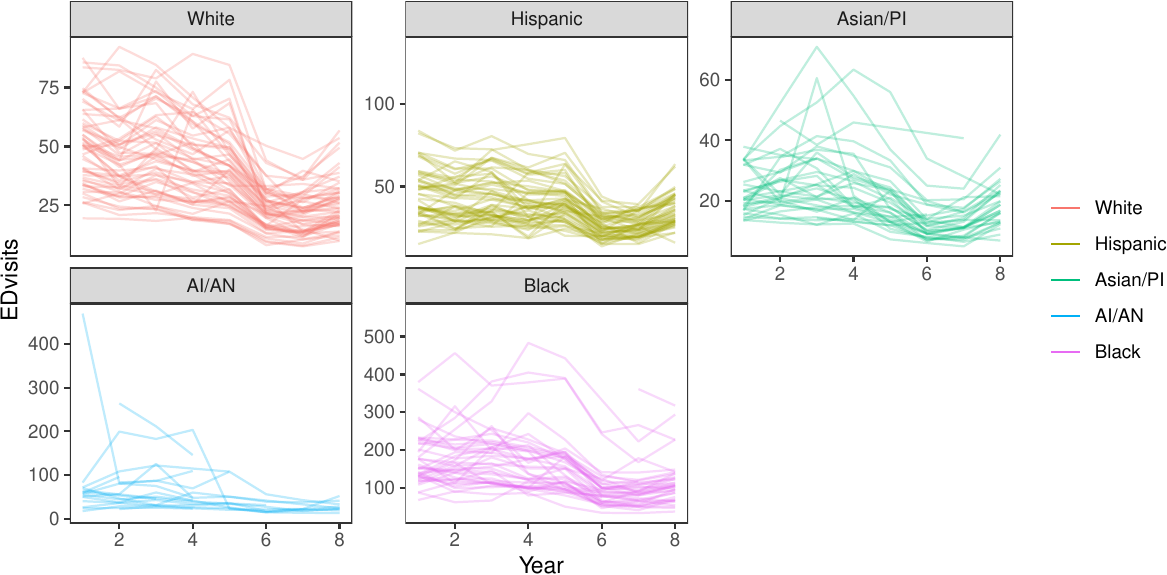}
    \caption{Annual trend in asthma-related ED visit rates per 10,000 residents for different racial groups.}
    \label{fig:asthma-panel}
\end{figure}

Next, we provide justification for the Gaussian modeling assumptions on the log-transformed ED visit rates. As shown in Figure~\ref{fig:asthma-hist}A, the distribution of the raw ED visit rates exhibits strong positive skewness across all racial groups, with a long right tail driven by a small number of counties reporting exceptionally high rates. This pronounced skewness is inconsistent with the symmetry assumed under Gaussian models. Applying a log transformation helps stabilize the variance and mitigate the effect of extreme values, resulting in a distribution that is more symmetric and approximately Gaussian (see Figure~\ref{fig:asthma-hist}B). The log-transformed rates exhibit reduced variability across the range of values and are more amenable to modeling under Gaussian assumptions.
\begin{figure}[t]
    \centering
    \includegraphics[width=0.95\linewidth]{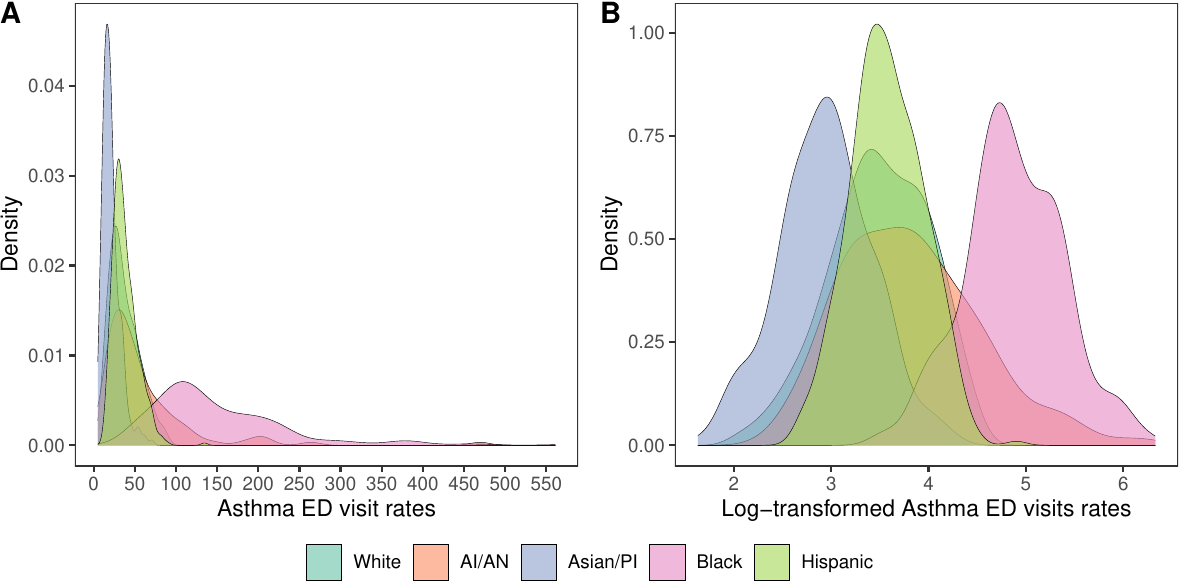}
    \caption{Race-specific estimated density of county-level age-adjusted rates (per 10,000) of asthma ED visits during 2015-2022.}
    \label{fig:asthma-hist}
\end{figure}

Further, we plot the recorded measurements of ozone levels across time and observe clear seasonal fluctuations and periodic patterns, indicative of strong temporal structure (see Figure~\ref{fig:oz-panel}). These recurring trends suggest the need to account for seasonality when modeling the temporal dynamics of ozone exposure. In addition to temporal variation, we also find substantial spatial heterogeneity in ozone concentrations, highlighting the importance of incorporating spatial effects into the modeling framework. Together, these observations motivate the use of flexible spatial-temporal models that can effectively capture both the periodic temporal behavior and the variability across different geographic locations.
\begin{figure}
    \centering
    \includegraphics[width=0.95\linewidth]{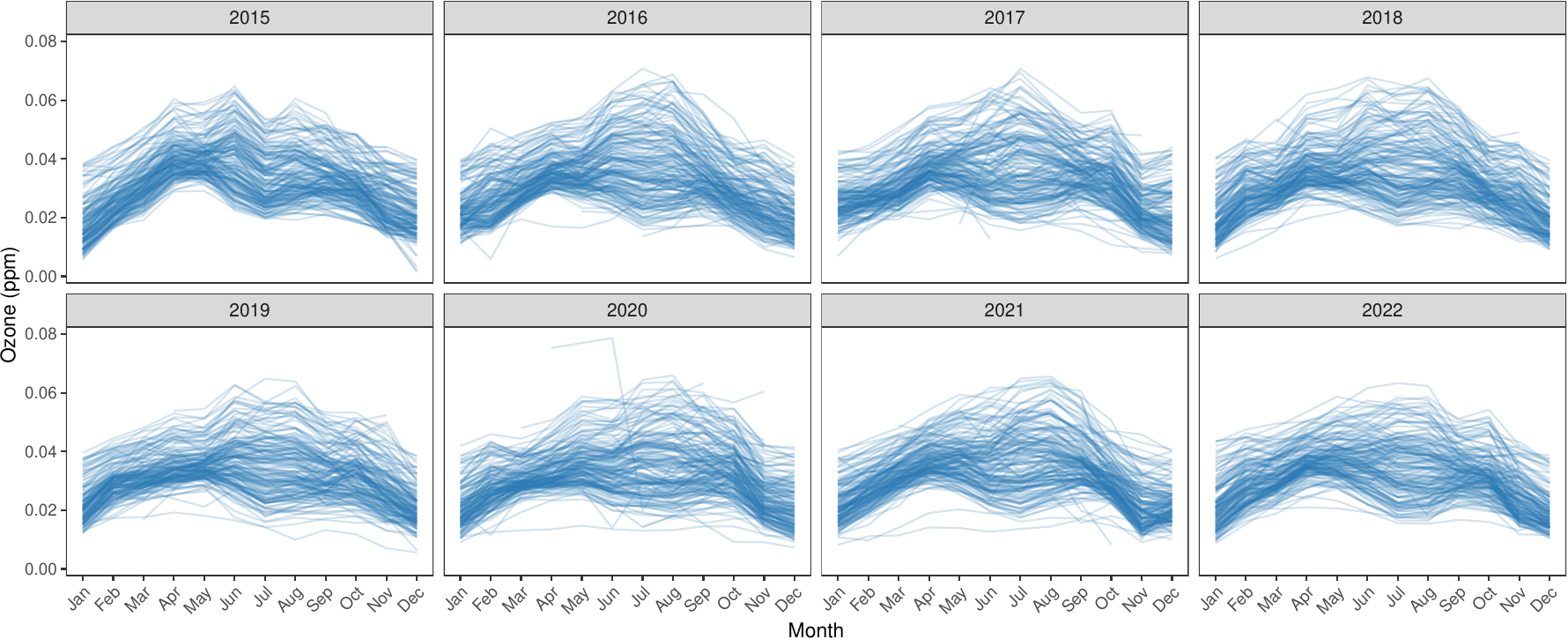}
    \caption{Monthly average ozone concentrations (in parts per million, ppm) recorded at various ozone monitoring sites across California from 2015 to 2022. A clear seasonal pattern is evident, with ozone levels peaking as well as exhibiting higher variability during the warmer months (May to October) and reaches troughs during the colder months (November to April).}
    \label{fig:oz-panel}
\end{figure}

\section{Additional simulation results}
Calculation of leave-one-out predictive densities is central to evaluating the objective function used to determine optimal stacking weights. In the main article, we present two approaches for computing LOO predictive densities. The first is an exact method based on closed-form expressions, while the second is an approximate method known as Pareto Smoothed Importance Sampling (PSIS). To assess the accuracy of PSIS relative to the exact method, we conduct a simulation study using a spatial regression model. Suppose $\{s_1, \ldots, s_n\}$ denote a collection of $n$ locations in the unit square, where we simulate spatially point-referenced responses $X(s_i)$ for $i = 1, \ldots, n$, using
\begin{equation}\label{eq:model1_supp}
    X(s_i) = W\beta + Z(s_i) + \epsilon(s_i)\;, \quad \epsilon(s_i) \ind \Norm (0, \delta^2 \sigma^2)\;,
\end{equation}
where $W$ is $n \times 2$ matrix comprising of an intercept and a predictor sampled from a standard normal distribution, $\beta = (2, 5)^{\T}$, $\epsilon(s_i)$ for each $i$ denote independent measurement error, $Z(s_i)$ are realizations of central spatial Gaussian process, given by $Z(s) \sim \GP (0, C_s(s, s'; \phi_s, \nu))$, with covariance function $C_s(\cdot, \cdot)$ as defined in \eqref{eq:corrfn}. We take $\sigma^2 = 0.4$, $\phi_s = 2$, $\nu = 0.5$ and $\delta^2 = 1.5$. 

We assign priors to $\beta$ and $\sigma^2$ as discussed in Section~\ref{subsec:candidate} of the main article. We fix the values of $\phi_s = 3.5$, $\nu = 0.75$, and $\delta^2 = 1$. Under this model and prior specification, we compute the exact leave-one-out predictive densities using closed-form expressions. To obtain the leave-one-out predictive densities via PSIS, we use posterior samples of $\beta$ and $Z = (Z(s_1), \ldots, Z(s_n))^{\T}$ obtained using the \textsf{R} package \texttt{spStack} \citep{spStack_r}, and calculate the log-pointwise predictive densities (lppd). We then apply the \texttt{psis()} function from the \textsf{R} package \texttt{loo} \citep{loo_pkg} to compute the leave-one-out predictive densities using stabilized importance weights. Figure~\ref{fig:exact-psis} presents a comparison of the two methods, showing no visible differences in the resulting predictive densities. However, the exact method is considerably more computationally intensive than PSIS, demonstrating the latter’s efficiency for large datasets.
\begin{figure}[t]
    \centering
    \includegraphics[width=0.5\linewidth]{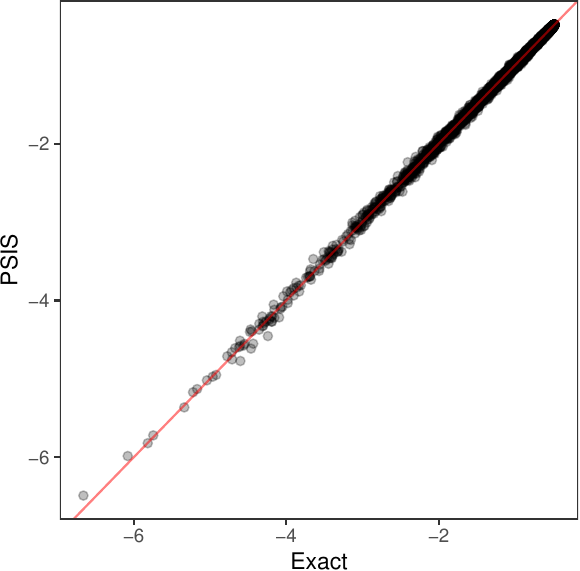}
    \caption{Comparison of leave-one-out predictive densities computed using exact closed form expression given by \eqref{eq:exactLOOPD} and Pareto smoothed importance sampling (PSIS) for a spatial regression model on a simulated dataset of sample size 5000.}
    \label{fig:exact-psis}
\end{figure}

\section{Technical details}\label{app:technical}
\renewcommand{\theproposition}{\arabic{proposition}}
\begin{proposition}
    For any $a<b$ and $c<d$, suppose $\tilde{C}_t(a, b, c, d; \phi) = \int_{c}^{d} \int_{a}^{b} C_t(t, t'; \phi_t) dt dt'$, then
    \begin{itemize}
        \item[(a)] for non-overlapping $(a, b)$ and $(c, d)$, with $a < b \leq c < d$,
        \[
        \tilde{C}_t(a, b, c, d; \phi_t) = \frac{1}{\phi_t^2} \left[ F(a, d) + F(b, c) - F(a, c) - F(b, d) \right]\;,
        \]
        \item[(b)] if $c = a$ and $d = b$, then
        \[
        \tilde{C}_t(a, b, a, b; \phi_t) = \frac{2}{\phi_t^2} \left[ \phi_t(b - a) + F(a, b) - 1 \right] \;,
        \]
        \item[(c)] if intervals $(a, b)$ and $(c, d)$ overlap, with $a \leq c < b < d$ or $a < c < b \leq d$,
        \[
        \tilde{C}_t(a, b, c, d; \phi_t) = \frac{1}{\phi_t^2} \left[ 2 \phi_t (b-c) + F(a, d) + F(c, b) - F(a, c) - F(b, d) \right]\;,
        \]
        \item[(d)] if $(a, b)$ is nested within $(c, d)$, i.e. either $c \leq a < b < d$ or $c < a < b \leq d$,
        \[
        \tilde{C}_t(a, b, c, d; \phi_t) = \frac{1}{\phi_t^2} \left[ 2\phi_t (b-a) + F(a, d) + F(c, b) - F(c, a) - F(b, d) \right]\;,
        \]
    \end{itemize}
    where $F(a_1, a_2) = F(a_1, a_2; \phi_t) = \exp \left( -\phi_t (a_2 - a_1) \right)$ for any $a_1, a_2 \in \calT$.
\end{proposition}
\begin{proof}
\begin{itemize}
\item[(a)] Without loss of generality, we assume $a<b<c<d$. The condition $a<b<c<d$ denotes that the intervals $(a, b)$ and $(c, d)$ are disjoint and $(a, b)$ lies on the left of $(c, d)$. Since, $t \in (a, b)$ and $t' \in (c, d)$. This means $t' \geq t$ holds true always, and hence, $|t - t'| = t' - t$. So, we evaluate the integral $I_1$ as follows.
\begin{equation*}
\begin{split}
    I_1 &= \int_{c}^{d} \int_{a}^{b} \exp(- \phi_t |t - t'|) dt dt'
    = \int_{c}^{d} \int_{a}^{b} \exp(- \phi_t (t' - t)) dt dt'\\
    &= \frac{1}{\phi_t} \left( e^{\phi_t b} - e^{\phi_t a} \right) \int_{c}^{d} \exp(- \phi_t t') dt'
    = \frac{1}{\phi_t^2} \left( e^{\phi_t b} - e^{\phi_t a} \right) \left(e^{-\phi_t c} - e^{- \phi_t d} \right)\\
    &= \frac{1}{\phi_t^2} \left( e^{\phi_t (b - c)} - e^{\phi_t (b-d)} - e^{\phi_t (a-c)} + e^{\phi_t (a-d)}\right)\\
    &= \frac{1}{\phi_t^2} \left[ F(b, c) - F(a, c) - F(b, d) + F(a, d) \right]
\end{split}
\end{equation*}
The equality case $b = c$ simply implies $F(c, b) = 0$.
\item[(b1)] If $c = a$ and $d = b$, then we are dealing with the integration
\begin{equation*}
\begin{split}
    \int_{a}^{b} \int_{a}^{b} e^{- \phi_t |t-t'|} dt dt' &= \int_{a}^{b} \left( \int_{a}^{t'} e^{- \phi_t |t-t'|} dt + \int_{t'}^{b} e^{- \phi_t |t-t'|} dt \right) dt' \\
    &= \int_{a}^{b} \left( \int_{a}^{t'} e^{- \phi_t (t'-t)} dt + \int_{t'}^{b} e^{- \phi_t (t-t')} dt \right) dt'\\
    &= \int_{a}^{b} \left( e^{- \phi_t t'} \cdot \frac{1}{\phi_t} \left(e^{\phi_t t'} - e^{\phi_t a} \right) + e^{\phi_t t'} \frac{1}{\phi_t} \left( e^{- \phi_t t'} - e^{- \phi_t b} \right) \right) dt'\\
    &= \frac{1}{\phi_t} \int_{a}^{b} 1 - e^{\phi_t (a - t')} + 1 - e^{\phi_t (t' - b)} dt'\\
    &= \frac{2}{\phi_t} (b - a) - \frac{1}{\phi_t} e^{\phi_t a} \int_{a}^{b} e^{- \phi_t t'} dt' - \frac{1}{\phi_t} e^{- \phi_t b} \int_{a}^{b} e^{\phi_t t'} dt'\\
    &= \frac{2}{\phi_t} (b - a) - \frac{1}{\phi_t^2} e^{\phi_t a} \left( e^{-\phi_t a} - e^{- \phi_t b} \right) - \frac{1}{\phi_t^2} e^{- \phi_t b} \left( e^{\phi_t b} - e^{\phi_t a} \right)\\
    &= \frac{2}{\phi_t} (b - a) - \frac{1}{\phi_t^2} \left( 1 - e^{\phi_t (a-b)} + 1 - e^{\phi_t (a-b)} \right)\\
    &= \frac{2}{\phi_t^2} \left( \phi_t(b - a) + e^{-\phi_t (b-a)} - 1 \right)\\
    &= \frac{2}{\phi_t^2} \left[ \phi_t(b - a) + F(a, b) - 1 \right]\;.
\end{split}
\end{equation*}

\item[(b2)] In this case, we have $a < c < b < d$. The integration relies on splitting the integral into parts depending on the sign of the term $|t-t'|$, where $t \in (a, b)$ and $t' \in (c, d)$. We split the integral into three double integrals $I_1$, $I_2$ and $I_3$, where $I_1$ is on $(c, d) \times (a, c)$, $I_2$ is on $(c, b) \times (c, b)$, and $I_3$ is on $(b, d) \times (c, b)$. 
\begin{equation*}
\begin{split}
    \int_{c}^{d} \int_{a}^{b} e^{- \phi_t |t - t'|} dt dt'
    &= \int_{c}^{d} \int_{a}^{c} e^{- \phi_t (t' - t)} dt dt' + \int_{c}^{b} \int_{c}^{b} e^{- \phi_t |t - t'|} dt dt' + \int_{b}^{d} \int_{c}^{b} e^{- \phi_t (t' - t)} dt dt'\\
    &= I_1 + I_2 + I_3\;,
\end{split}
\end{equation*}
In $I_1$, $t \leq t'$, so we have $|t - t'| = t' - t$. In $I_2$, $t - t'$ may be both positive and negative. And, for $I_3$, $t \leq t'$, so we have $|t - t'| = t' - t$. We tackle each of these integrals separately.
\begin{equation*}
\begin{split}
    I_1 &= \int_{c}^{d} \int_{a}^{c} e^{- \phi_t (t' - t)} dt dt'
    = \int_{c}^{d} e^{- \phi_t t'} \left(\int_{a}^{c} e^{\phi_t t} dt \right) dt'\\
    &= \frac{1}{\phi_t} \left( e^{\phi_t c} - e^{\phi_t a} \right) \int_{c}^{d} e^{- \phi_t t'} dt'
    = \frac{1}{\phi_t^2} \left( e^{\phi_t c} - e^{\phi_t a} \right) \left( e^{- \phi_t c} - e^{- \phi_t d} \right)\\
    &= \frac{1}{\phi_t^2} \left( 1 - e^{\phi_t (c-d)} - e^{\phi_t (a-c)} + e^{\phi_t (a - d)} \right)\;.
\end{split}
\end{equation*}
Similarly, we find the integral $I_3$, which is very similar to $I_1$.
\begin{equation*}
\begin{split}
    I_3 &= \int_{b}^{d} \int_{c}^{b} e^{- \phi_t (t' - t)} dt dt' = \int_{b}^{d} e^{- \phi_t t'} \left( \int_{c}^{b} e^{\phi_t t} dt \right) dt'\\
    &= \frac{1}{\phi_t^2} \left( e^{\phi_t b} - e^{\phi_t c} \right) \left( e^{- \phi_t b} - e^{- \phi_t d} \right) \\
    &= \frac{1}{\phi_t^2} \left( 1 - e^{\phi_t (b-d)} - e^{\phi_t (c-b)} + e^{\phi_t (c-d)} \right)
\end{split}
\end{equation*}
Following part~(b) of Proposition~\ref{prop:timeCov}, we have 
\begin{equation*}
\begin{split}
    I_2 = \int_{c}^{b} \int_{c}^{b} e^{- \phi_t |t - t'|} dt dt' = \frac{2}{\phi_t^2} \left( \phi_t(b - c) + e^{-\phi_t (b-c)} - 1 \right)
\end{split}
\end{equation*}
Combining $I_1$, $I_2$ and $I_3$, we have
\begin{equation*}
\begin{split}
    I_1 + I_2 + I_3 &= \frac{1}{\phi_t^2} \left( 1 - e^{\phi_t (c-d)} - e^{\phi_t (a-c)} + e^{\phi_t (a - d)} \right)
     + \frac{2}{\phi_t^2} \left( \phi_t(b - c) + e^{-\phi_t (b-c)} - 1 \right) \\
    & \quad + \frac{1}{\phi_t^2} \left( 1 - e^{\phi_t (b-d)} - e^{\phi_t (c-b)} + e^{\phi_t (c-d)} \right) \\
    &= \frac{1}{\phi_t^2} \left( 2 - e^{\phi_t (a-c)} + e^{\phi_t (a - d)} - e^{\phi_t (b-d)} - e^{\phi_t (c-b)}\right)\\
    & \quad + \frac{2}{\phi_t^2} \left( \phi_t(b - c) + e^{-\phi_t (b-c)} \right) - \frac{2}{\phi_t^2} \\
    &= \frac{2}{\phi_t} (b - c) + \frac{1}{\phi_t^2} \left(e^{\phi_t (c-b)}  - e^{\phi_t (a-c)} + e^{\phi_t (a - d)} - e^{\phi_t (b-d)}\right)\\
    &= \frac{1}{\phi_t^2} \left[ 2 \phi_t (b-c) + e^{- \phi_t (b-c)}  - e^{-\phi_t (c-a)} - e^{-\phi_t (d-b)} + e^{-\phi_t (d-a)} \right] \\
    &= \frac{1}{\phi_t^2} \left[ 2 \phi_t (b-c) + F(c, b) - F(a, c) - F(b, d) + F(a, d) \right]
\end{split}
\end{equation*}
\item[(c)] In this case the interval $(a, b)$ is nested within the interval $(c, d)$. Hence, we have $c<a<b<d$. We split the integral with respect to $t'$ on $(c, d)$ into three - one over $(c, a)$, one over $(a, b)$ and the last over $(b, d)$. Hence, we have
\begin{equation*}
\begin{split}
    \int_{c}^{d} \int_{a}^{b} e^{-\phi_t |t - t'|} dt dt' &=  \int_{c}^{a} \int_{a}^{b} e^{-\phi_t |t - t'|} dt dt' +  \int_{a}^{b} \int_{a}^{b} e^{-\phi_t |t - t'|} dt dt' +  \int_{b}^{d} \int_{a}^{b} e^{-\phi_t |t - t'|} dt dt'\\
    &= I_1 + I_2 + I_3\;.
\end{split}
\end{equation*}
We evaluate each of the integrals $I_1$, $I_2$ and $I_3$ separately. First, we find
\begin{equation*}
\begin{split}
    I_1 &= \int_{c}^{a} \int_{a}^{b} e^{-\phi_t (t - t')} dt dt' = \int_{c}^{a} e^{\phi_t t'} \left( \int_{a}^{b} e^{-\phi_t t} dt \right) dt'\\
    &= \frac{1}{\phi_t^2} \left( e^{-\phi_t a} - e^{-\phi_t b} \right) \left( e^{\phi_t a} - e^{\phi_t c} \right)\\
    &= \frac{1}{\phi_t^2} \left( 1 - e^{\phi_t (c-a)} - e^{\phi_t (a-b)} + e^{\phi_t (c-b)} \right)\;.
\end{split}
\end{equation*}
Similarly, we find the integral $I_3$ similar to $I_1$.
\begin{equation*}
\begin{split}
    I_3 &= \int_{b}^{d} \int_{a}^{b} e^{-\phi_t (t' - t)} dt dt' = \int_{b}^{d} e^{-\phi_t t'} \left( \int_{a}^{b} e^{\phi_t t} dt \right) dt'\\
    &= \frac{1}{\phi_t^2} \left( e^{\phi_t b} - e^{\phi_t a} \right) \left( e^{-\phi_t b} - e^{-\phi_t d} \right)\\
    &= \frac{1}{\phi_t^2} \left( 1 - e^{\phi_t (b-d)} - e^{\phi_t (a-b)} + e^{\phi_t (a-d)} \right)
\end{split}
\end{equation*}
Following part~(b) of Proposition~\ref{prop:timeCov}, we find the integral $I_2$. Next, we combine $I_1$, $I_2$ and $I_3$.
\begin{equation*}
\begin{split}
    I_1 + I_2 + I_3 &= \frac{1}{\phi_t^2} \left( 1 - e^{\phi_t (c-a)} - e^{\phi_t (a-b)} + e^{\phi_t (c-b)} \right) + \frac{2}{\phi_t^2} \left( \phi_t(b - a) + e^{-\phi_t (b-a)} - 1 \right)\\
    & \quad + \frac{1}{\phi_t^2} \left( 1 - e^{\phi_t (b-d)} - e^{\phi_t (a-b)} + e^{\phi_t (a-d)} \right)\\
    &= \frac{1}{\phi_t^2} \left( 2 - 2 e^{\phi_t (a-b)} - e^{\phi_t (c-a)} + e^{\phi_t (c-b)} - e^{\phi_t (b-d)} + e^{\phi_t (a-d)} \right) \\
    & \quad + \frac{2}{\phi_t} (b-a) + \frac{1}{\phi_t^2} \left( 2 e^{\phi_t (a-b)} - 2\right)\\
    &= \frac{1}{\phi_t^2} \left[ 2\phi_t (b-a) + e^{-\phi_t (b-c)} - e^{-\phi_t (a-c)} - e^{-\phi_t (d-b)} + e^{-\phi_t (d-a)} \right]\\
    &= \frac{1}{\phi_t^2} \left[ 2\phi_t (b-a) + F(c, b) - F(c, a) - F(b, d) + F(a, d) \right]
\end{split}
\end{equation*}
\end{itemize}
\end{proof}

\end{document}